\documentclass[sigconf]{acmart}

\AtBeginDocument{%
  }

\copyrightyear{2026}
\acmYear{2026}
\setcopyright{cc}
\setcctype{by-nc-nd}
\acmConference[MM '26] {Proceedings of the 34th ACM International Conference on Multimedia}{November 10--14, 2026}{Rio de Janeiro, Brazil.}
\acmBooktitle{Proceedings of the 34th ACM International Conference on Multimedia (MM '26), November 10--14, 2026, Rio de Janeiro, Brazil}

\acmISBN{979-8-4007-2213-4/2026/11}
\acmDOI{10.1145/3767308.3835863}

\usepackage{bbding}
\usepackage{amsfonts}

\usepackage{amssymb}

\usepackage{bm}
\usepackage{multirow}

\usepackage{rotating}
\usepackage{amsthm}
\usepackage{color}
\usepackage{stfloats}
\newtheorem{proposition}{Proposition}

\usepackage{colortbl}
\usepackage{pifont}
\usepackage{textcomp}
\usepackage{url}

\hypersetup{hidelinks}
\usepackage[noend]{algorithmic}
\usepackage{enumitem}

\usepackage{bbding}
\usepackage{amsfonts}
\usepackage{amsmath}

\usepackage{amssymb}

\usepackage{graphicx}
\usepackage{float}
\usepackage{bm}
\usepackage{courier}
\usepackage{booktabs}
\usepackage{multirow}
\usepackage{rotating}
\usepackage{amsthm}
\usepackage{footnote}
\usepackage{color}
\usepackage{stfloats}
\usepackage{threeparttable}
\usepackage{adjustbox}

\usepackage{orcidlink}
\usepackage{colortbl}
\usepackage{pifont}
\usepackage{textcomp}
\usepackage{url}

\hypersetup{hidelinks}
\usepackage[noend]{algorithmic}
\usepackage{xcolor}
\usepackage{pgfplots}
\usepgfplotslibrary{groupplots}
\pgfplotsset{compat=1.18}
\usepackage{enumitem}
\usepackage{hyperref}

\usepackage[most]{tcolorbox}

\begin{document}

\title{Concept-Level Risk and Calibration for Governance in Diffusion Foundation Models}



\author{Kun Xu}
\affiliation{%
  \institution{Nanjing University of Aeronautics and Astronautics}
  \city{Nanjing}
  \country{China}}
\email{xukun930@nuaa.edu.cn}

\author{Yushu Zhang}
\authornote{Corresponding author}
\affiliation{%
  \institution{Jiangxi University of Finance and Economics}
  \city{Nanchang}
  \country{China}}
\email{zhangyushu@jxufe.edu.cn}

\author{Tao Wang}
\affiliation{%
  \institution{Nanjing University of Aeronautics and Astronautics}
  \city{Nanjing}
  \country{China}}
\email{wangtao21@nuaa.edu.cn}

\author{Shuren Qi}
\affiliation{%
  \institution{City University of Hong Kong}
  \city{Hong Kong}
  \country{China}}
\email{shurenqi@cityu.edu.hk}

\author{Barbara Carminati}
\affiliation{%
  \institution{University of Insubria}
  \city{Varese}
  \country{Italy}}
\email{barbara.carminati@uninsubria.it}

\author{Elena Ferrari}
\affiliation{%
  \institution{University of Insubria}
  \city{Varese}
  \country{Italy}}
\email{elena.ferrari@uninsubria.it}

\author{Yuming Fang}
\affiliation{%
  \institution{Jiangxi University of Finance and Economics}
  \city{Nanchang}
  \country{China}}
\email{fa0001ng@e.ntu.edu.sg}

\renewcommand{\shortauthors}{Kun Xu et al.}

\begin{abstract}
Diffusion models have become a core paradigm for multimedia generation, offering powerful concept-driven controllability for personalization, semantic editing, and selective unlearning. However, as semantic control extends beyond natural-language prompts to learned embeddings and intervention pipelines, the safety and governance of these systems become increasingly difficult to evaluate in a unified manner, especially for safety-sensitive, identity-linked, and other privacy-relevant concepts. Existing studies mainly rely on heuristic audits, adversarial probing, or task-specific erasure benchmarks, and therefore provide limited support for systematic comparison across models, conditioning channels, and deployment conditions. We present a concept-level probabilistic audit and reporting framework for diffusion models. We formalize governance-relevant concept behaviors as Bernoulli semantic events induced by stochastic generation, and define a Concept Risk Operator that maps model-channel configurations to structured risk profiles, enabling comparison across prompting interfaces, learned embedding channels, models, and recorded conditions. We apply sample-level post-hoc calibration and configuration-level risk aggregation, and show that probability error can change thresholded actions near policy boundaries. Experiments on SD1.5, SD2.1, and SDXL reveal consistent yet non-uniform operational risk patterns across concept families, channels, recorded conditions, and shifted protocols. In particular, embedding-based access and obfuscated prompts expose risks often understated by standard-prompt evaluation. A pooled multi-protocol calibrator improves held-out probability reliability, but we do not claim transfer from a standard-only calibrator. CLRC provides a common audit schema for probabilistic and decision-aware governance of multimedia generation systems.
\end{abstract}




\maketitle

\section{Introduction}
Diffusion models have rapidly evolved into powerful visual foundation models, achieving remarkable performance in high-fidelity image synthesis and semantic controllability~\cite{10419041, ma2025efficient}. Following the introduction of denoising diffusion probabilistic models~\cite{ho2020ddpm} and latent diffusion models~\cite{rombach2022ldm}, text-to-image (T2I) systems such as Stable Diffusion (SD) have demonstrated strong alignment between textual descriptions and generated images~\cite{10230895, 10489849}. Subsequent advances in personalization and concept manipulation, including Textual Inversion~\cite{gal2022textual}, DreamBooth~\cite{ruiz2023dreambooth}, and cross-attention-based editing~\cite{hertz2022prompttoprompt}, have transformed diffusion models from generic generative tools into programmable semantic systems capable of concept injection, modification, and removal~\cite{10646735, guo2025conceptguard, kumari2023concept}.

While these capabilities substantially enhance controllability, they simultaneously introduce new governance and reliability challenges~\cite{THEMIS2025ndss}. In concept-driven diffusion settings, semantic concepts act as operational units that can be activated, suppressed, transferred, or combined. Sensitive concepts such as weapons, nudity, copyrighted artistic styles, identity-related attributes, or other privacy-relevant semantic cues may be intentionally erased through model editing~\cite{gandikota2023erasing, kumari2023concept, lu2024mace}, yet remain recoverable through alternative conditioning channels or adaptive prompts~\cite{petsiuk2024concept, chen2025mind}. This issue is particularly important for identity-linked and personalized generation settings, where semantic control may interact with subject-specific embeddings and thereby expand the effective governance surface beyond prompt-only access. Moreover, robustness-oriented defenses~\cite{NEURIPS2024_40954ac1, kim2024race} often target specific attack surfaces without providing a holistic view of concept-level vulnerabilities across models, conditioning channels, and deployment configurations.

Current research primarily addresses complementary subsets of this problem. Concept erasure and unlearning methods aim to suppress targeted concepts~\cite{gandikota2023erasing, kumari2023concept}, while red-teaming approaches evaluate prompt-based safety weaknesses~\cite{tsai2024ringabell}. Improvements in diffusion guidance and sampling stability focus on generation quality or classifier alignment~\cite{nichol2021improved, dhariwal2021guided}. Benchmarks such as Six-CD cover important concept-removal settings~\cite{ren2025six}; what remains less standardized is joint reporting across model, channel, protocol distribution, intervention state, benign loss, and probability calibration.

This gap becomes particularly critical when diffusion models are deployed as foundation systems. In real-world usage, concept activation behavior depends not only on textual prompts but also on learned embeddings introduced through personalization techniques~\cite{gal2022textual, ruiz2023dreambooth, richardson2024conceptlab}. This channel expansion is governance-relevant not only for unsafe content, but also for identity-linked and privacy-relevant concepts, since subject-specific embeddings or alternate control pathways may preserve, recover, or amplify concept behaviors that are not visible under prompt-only evaluation~\cite{dubinski2025cdi}. The same semantic concept may exhibit distinct reproduction probabilities under different architectures such as SD1.5 and SDXL, as well as under different conditioning channels. Furthermore, existing evaluation metrics typically report binary outcomes or aggregate success rates, without quantifying whether predicted risk aligns with empirical frequency. Classical calibration theory~\cite{guo2017calibration} suggests that reliable decision-making requires probabilistic estimates whose confidence reflects true event likelihood, yet such calibration analysis has not been systematically applied to semantic concept events in diffusion models. These observations motivate three fundamental research questions:
\begin{itemize}[leftmargin=*]
\item \textbf{RQ1:} How does semantic risk migrate across prompt, embedding, and intervention pathways under a unified configuration space?
\item \textbf{RQ2:} Can concept activation be modeled as a Bernoulli semantic event, with reproduction and bypass as comparable risks and spillover as a population-level collateral-risk diagnostic?
\item \textbf{RQ3:} How do held-out probability errors and post-hoc calibration affect sample-level threshold actions near policy boundaries?
\end{itemize}

To address these questions, we present \emph{Concept-Level Risk and Calibration} (CLRC) as an audit and reporting schema rather than a new concept detector, erasure algorithm, or calibration method. CLRC indexes event-frequency estimators by concept, model, channel, protocol distribution, and intervention state, and couples these views with held-out post-hoc calibration using a declared multi-protocol development/held-out design. Its intended contribution is coverage, comparability, and decision-oriented reporting. We instantiate this schema in a controlled protocol for cross-model, cross-channel, and intervention-aware analysis. Empirically, the protocol reveals structured concept-level risk patterns, cross-channel risk shifts, and systematic miscalibration in concept-risk estimation. Figure~\ref{fig:overview} provides an overview of the CLRC pipeline, from concept-conditioned generation and event construction to structured risk estimation and post-hoc calibration for governance. Our contributions are summarized as follows:
\begin{itemize}[leftmargin=*]
\item \textit{Unified concept-risk evaluation across access and governance pathways.} Existing diffusion safety evaluations are mostly prompt-centric. We formalize an audit operator over model, channel, protocol distribution, and condition state, enabling direct comparison across prompt access, learned embedding access, and recorded pre/post-condition settings.

\item \textit{Bernoulli semantic-event modeling for concept-level governance.} Rather than reporting only binary success or suppression rates, we model reproduction and bypass as Bernoulli concept events and treat spillover as a population-level collateral-risk diagnostic. This provides a common estimation target for concept-level governance. 

\item \textit{Calibration-aware analysis of governance stability.} We show that miscalibration can change thresholded actions near policy boundaries. Our experiments quantify held-out probability reliability and action sensitivity to post-hoc calibration; they do not assert standard-to-shifted transfer or oracle decision accuracy.
\end{itemize}
\section{Related Work}\label{sec-rw}

\subsection{Concept Conditioning, Editing, and Erasure in Diffusion Models} \label{sec-rw-diffusion-control}

Diffusion models have become a dominant class of foundation generators for visual synthesis due to scalable training and high-fidelity sampling~\cite{fuest2026diffusion, liu2026alignment}. DDPM establishes the forward noising and reverse denoising framework~\cite{ho2020ddpm}, while subsequent advances improve objectives and sampling efficiency. Latent diffusion substantially reduces computation by performing generation in a learned latent space without sacrificing quality~\cite{rombach2022ldm}. Diffusion architectures have also evolved beyond U-Nets, with transformer-based backbones showing favorable scaling behavior~\cite{peebles2023dit}. Large-scale systems further benefit from stronger text encoders and guidance mechanisms: Imagen highlights the role of large language models as text encoders~\cite{saharia2022imagen}, and classifier-free guidance enables controllable fidelity--diversity trade-offs without a separate classifier~\cite{ho2022cfg}. More broadly, score-based formulations provide an explicit stochastic foundation for diffusion-style generation~\cite{song2021scorebased, ding2026ccdm}, while recent surveys review the growing landscape of controllable T2I diffusion systems~\cite{Controllable2025}.

As diffusion models mature into concept-driven foundation systems, practical utility increasingly depends on the ability to inject, bind, edit, suppress, and recover semantic concepts through diverse conditioning channels. Personalization methods bind new concepts or subjects to tokens or embeddings: Textual Inversion learns a new token embedding while freezing the base model~\cite{gal2022textual}, whereas DreamBooth fine-tunes the model to associate a rare token with a specific subject under prior-preservation regularization~\cite{ruiz2023dreambooth}. At inference time, concept manipulation can be achieved through attention and inversion mechanisms. Prompt-to-Prompt localizes word-level effects through cross-attention control~\cite{hertz2022prompttoprompt}, and Null-text inversion improves real-image reconstruction and editing by optimizing the unconditional embedding used in classifier-free guidance~\cite{mokady2023nulltext}. Instruction-based editing further extends prompt control to natural-language edit commands, as in InstructPix2Pix~\cite{brooks2023instructpix2pix}, while ControlNet adds spatial or structural conditions such as edges, depth, pose, or segmentation without retraining the backbone~\cite{zhang2023controlnet}.

A closely related line studies concept manipulation from editing to removal and forgetting~\cite{lin2025ice}. Early editing frameworks such as SDEdit steer the denoising process under conditional priors to balance realism and fidelity~\cite{meng2021sdedit}. Training-free or lightweight methods support semantic editing through prompt control, discovered embedding directions, or inferred masks, including Prompt-to-Prompt~\cite{hertz2022prompttoprompt}, Null-text inversion~\cite{mokady2023nulltext}, InstructPix2Pix~\cite{brooks2023instructpix2pix}, pix2pix-zero~\cite{parmar2023pix2pixzero}, and DiffEdit~\cite{couairon2022diffedit}. Beyond editing, concept erasing and unlearning aim to suppress specific concepts, styles, or unsafe content while preserving model utility~\cite{li2025set, chen2025trce, gong2024reliable}. Safe Latent Diffusion introduces inference-time interventions for suppressing unsafe generation without retraining~\cite{schramowski2023sld}; ESD performs weight-level concept removal through negative-guidance fine-tuning~\cite{gandikota2023erasing}; and UCE uses closed-form parameter editing for moderation, debiasing, and style erasure at scale~\cite{gandikota2024uce}. Machine unlearning surveys further systematize relevant definitions, threat models, and algorithm families~\cite{qi2025forget, NguyenunlearningSurvey, gao2025meta}, while Forget-Me-Not and SalUn study selective forgetting and saliency-based removal in diffusion and broader generative settings~\cite{zhang2023forgetmenot,fan2024salun, rusanovsky2025memories}. Recent analyses show that concept removal is difficult to validate and may remain brittle under adversarial prompts or alternate channels~\cite{liu2025multimodal}: Ring-A-Bell reveals model-agnostic bypass behavior~\cite{tsai2024ringabell}, and subsequent work examines whether concepts are truly erased and what collateral side effects remain~\cite{lu2025whenErased}. Together, these studies motivate the development of concept-level risk evaluation frameworks for systematically comparing editing, erasure, bypass, and spillover behaviors across models and conditioning channels.

Embedding-space vulnerability predates CLRC: query-free attacks against Stable Diffusion exploit sensitive text-encoder dimensions~\cite{zhuang2023pilot}. Recent defenses include HiRM, which redirects high-level representations~\cite{lee2026hirm}, and AEGIS, which uses adversarial targets without retention data~\cite{li2026aegis}. Such methods can instantiate a fully specified intervention $\pi$; CLRC is complementary rather than a competing optimizer, auditing reproduction, bypass, benign loss, and calibration under matched protocols.
\begin{figure*}[t]
\centering
\includegraphics[width=\textwidth]{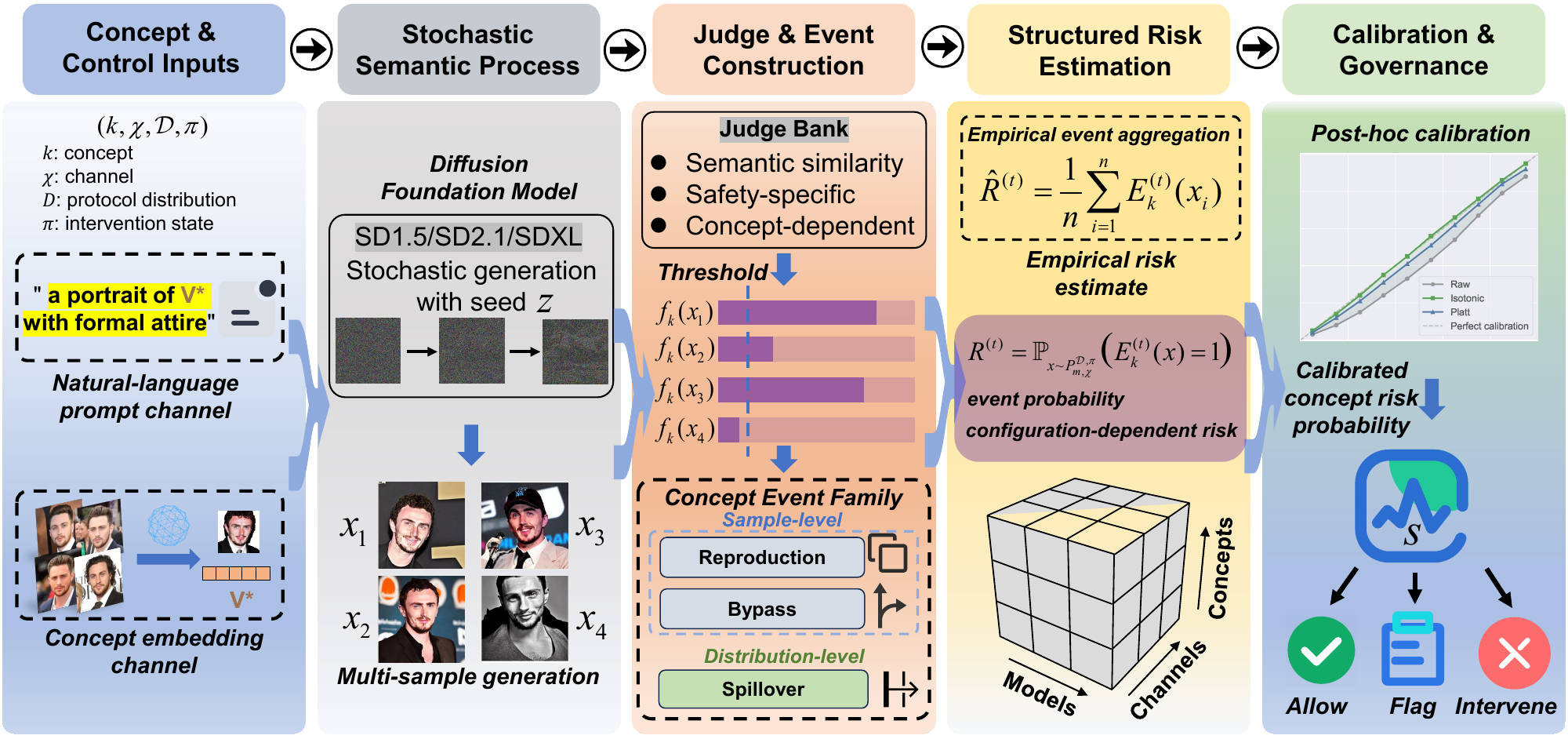}
\caption{Overview of Concept-Level Risk Modeling and Calibration in Diffusion Models.} 
\Description{A five-stage pipeline. Concept and control inputs feed a stochastic diffusion generator; a fixed judge constructs thresholded concept events; event frequencies are aggregated into structured risks over models, channels, and concepts; and post-hoc calibration supports allow, flag, or intervene actions.}
\label{fig:overview}
\end{figure*}

\subsection{Safety and Governance Challenges in Diffusion Models} \label{sec-rw-safety-governance}

As diffusion models are increasingly deployed, safety evaluation and governance have become central concerns. Imagen introduced DrawBench for human evaluation of compositional and prompt-following behavior~\cite{saharia2022imagen}; T2ISafety audits fairness, toxicity, and privacy~\cite{li2025t2isafety}, while Six-CD and UnlearnCanvas benchmark concept-removal effectiveness and retainability~\cite{ren2025six,zhang2024unlearncanvas}. These task-specific benchmarks provide valuable evaluation targets. CLRC contributes an orthogonal reporting layer: it indexes reproduction, bypass, and benign loss jointly by model, channel, protocol, and intervention, and adds held-out human-referenced calibration and action-sensitivity analysis. We therefore claim broader joint reporting, not superior detector or erasure performance. Inference-time methods such as Safe Latent Diffusion provide safety guidance and curated prompt testbeds without retraining~\cite{schramowski2023sld}. At the same time, empirical evidence suggests that safety filters and alignment layers can be brittle under distribution shifts or adversarial prompting~\cite{villa2025exposing, wu2025proactive}, indicating that evaluation should extend beyond nominal prompt settings to multiple channels and attack surfaces.

Beyond content moderation, diffusion models also raise privacy, memorization, fairness, and governance concerns. Web-scale training may induce memorization and reproduction of training samples; Carlini et al.\ show that diffusion models can leak verbatim training images under targeted extraction~\cite{carlini2023extracting}, while membership inference attacks indicate broader training-data exposure risks~\cite{shokri2017membership}. Open datasets such as LAION-5B further introduce noise and potential demographic or cultural bias into downstream diffusion systems~\cite{schuhmann2022laion5b}. Existing mitigation strategies often rely on curation, filtering, debiasing, or adversarial testing, but they typically focus on isolated failure modes and provide limited support for structured comparison across concept families, models, and interfaces. Calibration has long been recognized as critical for safety-critical decision making, since miscalibration can destabilize threshold-based interventions near policy boundaries~\cite{vaicenavicius2019evaluating}. However, most diffusion safety evaluations still rely on heuristic auditing or aggregate success rates, without modeling governance-relevant concept activation as calibrated Bernoulli events. This leaves a gap between empirical auditing and decision-oriented concept-level risk assessment, which our work aims to address.

\section{Concept Risk Modeling and Calibration}\label{sec-CRMC}

\subsection{Concept Event Family} \label{sec-event-family}
We model a diffusion foundation model as a conditional stochastic generator over an image space $\mathcal{X}$. Let $m\in\mathcal{M}$ index the model architecture and let $\chi\in\{\mathrm{prompt},\mathrm{embedding}\}$ denote the input channel. Let $z\sim p(z)$ be the latent noise seed and let $c_{\mathrm{cond}}\sim\mathcal{D}$ be a conditioning input drawn from a distribution $\mathcal{D}$ under channel $\chi$. The generation process is $x = G_m\!\left(z, c_{\mathrm{cond}}^{(\chi)}\right)$, which induces a probability measure over $\mathcal{X}$ given by
\begin{equation}
P_{m,\chi}^{\mathcal{D}} = \mathrm{Law}\!\left(G_m\!\left(z, c_{\mathrm{cond}}^{(\chi)}\right)\right), \qquad z\sim p(z),\; c_{\mathrm{cond}}\sim\mathcal{D}. \label{eq-induced-distribution} \end{equation}

Here, $\mathrm{Law}(\cdot)$ denotes the distribution of its random argument, induced jointly by the sampled seed and conditioning input. Thus each $(m,\chi,\mathcal D)$ defines an image distribution. A governance condition $\pi$ induces $P_{m,\chi}^{\mathcal D,\pi}=\Pi_\pi(P_{m,\chi}^{\mathcal D})$, where $\Pi_\pi$ replaces the baseline pipeline rather than estimating risk.

We formalize governance-relevant behavior as measurable concept events. Let $\mathcal K=\{1,\dots,K\}$. For each $k\in\mathcal K$, let $Y_k(x)\in\{0,1\}$ denote the human-reference event, observed here only through archived binary annotations on the audited calibration set. A fixed judge gives $s_k(x)=f_k(x)$ and the operational event $E_k(x)=\mathbb I\{s_k(x)\ge\tau_k\}$. The risk tensor uses $E_k$, whereas calibration maps a fixed score proxy to $Y_k$; $Y_k$ is a reference target, not latent truth or an oracle.

The sample-level type $t\in\mathcal T=\{\mathrm{rep},\mathrm{byp}\}$ restricts admissible protocol--condition pairs without changing the judge. Let $\mathcal C_{\mathrm{rep}}$ contain the declared standard, shifted, and obfuscated pre-condition reproduction slices and the matched standard post-condition residual slice; let $\mathcal C_{\mathrm{byp}}$ contain the standard and obfuscated post-condition attack slices. For $(\mathcal D,\pi)\in\mathcal C_t$, set $E_k^{(t)}=E_k$ and $Y_k^{(t)}=Y_k$ under $P_{m,\chi}^{\mathcal D,\pi}$. Thus the types differ by admissible slices, not decision rules; operational bypass does not imply configuration-invariant judge error.

Interventions may also affect benign concepts. For $b\in\mathcal K_{\mathrm{ben}}$, let $R^{(\mathrm{rep})}(b,\cdot)$ denote the analogous fixed-judge event frequency; define the one-sided benign-utility loss and its tolerance indicator as
\begin{equation}
\resizebox{\linewidth}{!}{%
$\begin{aligned}
D_b(m,\chi;\mathcal{D},\pi)
&=\Big[R^{(\mathrm{rep})}(b,m,\chi;\mathcal{D},\varnothing)
-R^{(\mathrm{rep})}(b,m,\chi;\mathcal{D},\pi)\Big]_+,\\
\mathsf{S}_b(\epsilon)&=\mathbb{I}\{D_b\ge\epsilon\},
\qquad b\in\mathcal{K}_{\mathrm{ben}} .
\end{aligned}$%
}
\label{eq-spillover}
\end{equation}
Here $[u]_+=\max(u,0)$. We report the continuous benign-loss magnitude $\overline D(m,\chi;\mathcal D,\pi)=|\mathcal{K}_{\mathrm{ben}}|^{-1}\sum_bD_b(m,\chi;\mathcal D,\pi)$ and, when thresholding is required, the rate $|\mathcal{K}_{\mathrm{ben}}|^{-1}\sum_b\mathsf{S}_b(\epsilon)$. In our protocol $\mathcal{K}_{\mathrm{ben}}$ is a disjoint $20$-concept auxiliary benign set and $\epsilon=0.05$; Fig.~\ref{fig:comparative_multi}(d) reports one $\overline D$ for each displayed model--channel--protocol--condition configuration, not a target-family-specific statistic. Both quantities are population-level diagnostics and are excluded from $\mathcal{T}$.

\subsection{Structured Risk Profiling} \label{sec-structured-risk}

The Concept Risk Operator induces a structured risk profile, which we also view as a risk tensor indexed by concept, event type, model, channel, protocol distribution, and intervention. We formalize \emph{structured risk profiling} as a distribution-level characterization of governance-relevant concept behavior across these axes. For $t\in\mathcal T$, $k\in\mathcal K$, and an admissible pair $(\mathcal D,\pi)\in\mathcal C_t$, we define the event risk under model--channel configuration $(m,\chi)$ as the population-level probability
\begin{equation}
\begin{aligned}
R^{(t)}(k,m,\chi;\mathcal{D},\pi) &\triangleq \mathbb{P}_{x\sim P_{m,\chi}^{\mathcal{D},\pi}} \big(E_k^{(t)}(x)=1\big) = \mathbb{E}_{x\sim P_{m,\chi}^{\mathcal{D},\pi}} \big[E_k^{(t)}(x)\big], \end{aligned}
\label{eq-risk-event-type}
\end{equation}
where $\pi=\varnothing$ denotes the nominal setting. Only declared pairs $(\mathcal D,\pi)\in\mathcal C_t$ are admissible for event type $t$. This definition covers reproduction and bypass risks; spillover is handled separately as a distribution-level diagnostic in Eq.~\eqref{eq-spillover}.

Collecting risks across concepts yields the \emph{event-specific risk vector}
\begin{equation}
\resizebox{\linewidth}{!}{%
$\begin{aligned}
\mathbf{R}^{(t)}(m,\chi;\mathcal{D},\pi)
&= \big(
R^{(t)}(1,m,\chi;\mathcal{D},\pi), \dots, 
R^{(t)}(K,m,\chi;\mathcal{D},\pi)
\big)
\in [0,1]^K .
\end{aligned}$%
}
\label{eq-risk-vector-event}
\end{equation}
Collecting only admissible entries defines the structured risk profile
\begin{equation}
\mathbf{R}(m,\chi)=\left\{\mathbf{R}^{(t)}(m,\chi;\mathcal D,\pi):
t\in\mathcal T,\;(\mathcal D,\pi)\in\mathcal C_t\right\},
\label{eq-structured-risk}
\end{equation}
indexed by concept, event type, model, channel, protocol distribution, and condition.

Fix $(k,m,\mathcal{D},\pi,t)$. For a change from channel $\chi_1$ to $\chi_2$, define the signed channel risk gap
\begin{equation}
\begin{aligned}
\Delta^{(t)}_k(m;\chi_1\!\rightarrow\!\chi_2)
&=
R^{(t)}(k,m,\chi_2;\mathcal{D},\pi)
- R^{(t)}(k,m,\chi_1;\mathcal{D},\pi),
\end{aligned}
\label{eq-channel-gap-structured}
\end{equation}
and for a change from model $m_1$ to $m_2$ under fixed $(k,\chi,\mathcal{D},\pi,t)$, the model risk gap is
\begin{equation}
\begin{aligned}
\Gamma^{(t)}_k(\chi;m_1\!\rightarrow\!m_2)
&=
R^{(t)}(k,m_2,\chi;\mathcal{D},\pi)
- R^{(t)}(k,m_1,\chi;\mathcal{D},\pi).
\end{aligned}
\label{eq-model-gap-structured}
\end{equation}
Thus a positive prompt-to-embedding gap means higher fixed-judge frequency under embedding access, while a negative SD1.5-to-SDXL gap means a decrease on SDXL. Under configuration-invariant error with $\alpha_k+\beta_k<1$, these gaps preserve human-reference ordering; configuration-dependent error need not.

Under independent sampling, risks are estimated from $\{x_i\}_{i=1}^n\sim P_{m,\chi}^{\mathcal{D},\pi}$ by
\begin{equation}
\hat{R}^{(t)}(k,m,\chi;\mathcal{D},\pi) = \frac{1}{n} \sum_{i=1}^{n} E_k^{(t)}(x_i)
\label{eq-empirical-risk-structured}
\end{equation}
and is unbiased and consistent for $R^{(t)}(k,m,\chi;\mathcal{D},\pi)$. Our matched fixed-seed audit estimates the corresponding finite-protocol frequency and uses the same seed set across paired configurations; uncertainty is therefore interpreted conditional on that audit design.

\subsection{Probabilistic Estimation and Calibration} \label{sec-prob-calibration}

Structured risk profiling defines population-level event frequencies under model-induced distributions. Calibration is performed at the sample level and then aggregated, rather than by training a separate neural risk network. Let $j$ index a score--label pair, with concept $k_j$, channel $\chi_j$, human annotation $y_j=Y_{k_j}(x_j)$, and cosine score $s_j=f_{k_j}(x_j)$. We use the fixed raw probability proxy $r_j=\operatorname{clip}_{[0,1]}(s_j)$ and fit a channel-level post-hoc map
\begin{equation}
\begin{gathered}
q_j
= h_{\phi,\chi_j}(r_j)
\approx
\mathbb{P}\!\left(Y_{k_j}=1 \mid r_j,\chi_j\right),
\\
\widehat R_{\mathrm{cal}}(k,c)
=
\frac{1}{|\mathcal I_{k,c}|}
\sum_{j\in\mathcal I_{k,c}} q_j .
\end{gathered}
\label{eq-calibrated-risk}
\end{equation}
Here $c=(t,m,\chi,\mathcal D,\pi)$ and $\mathcal I_{k,c}=\{j:k_j=k,\ c_j=c\}$ is the concept-specific held-out subset. The two maps $h_{\phi,\chi}$ pool development pairs across concepts and protocol slices, but aggregation retains the concept index. This human-reference probability aggregate differs from the operational judge-event frequency in Eq.~\eqref{eq-empirical-risk-structured}. For probability-reliability evaluation, we additionally define $\mathcal I_c^{\mathrm{pool}}=\bigcup_k\mathcal I_{k,c}$ within a named held-out slice; the reported ECE/Brier values are pooled sample-level diagnostics, not concept-specific policy risks.

Threshold fitting and probability calibration use the development annotations for different mappings: $\tau_k$ converts $s$ into the operational event $E$, whereas $h_{\phi,\chi}$ maps the fixed proxy $r$ to the human-reference target $Y$. Neither $E$ nor a calibrated output is reused as its own calibration target. A principled sample-level measure of probabilistic accuracy is the Brier score,
\begin{equation}
\mathrm{BS} = \frac{1}{N_{\mathrm{cal}}}\sum_{j=1}^{N_{\mathrm{cal}}}
\big(q_j-y_j\big)^2,
\label{eq-brier}
\end{equation}
computed only on held-out human-labeled score--label pairs.

Calibration evaluates whether predicted probabilities agree with empirical human-label frequencies. Partitioning $[0,1]$ into bins $\{\mathcal{C}_b\}_{b=1}^B$ and letting $\mathcal{I}_b=\{j:q_j\in\mathcal{C}_b\}$, with $N_{\mathrm{cal}}=\sum_b|\mathcal{I}_b|$, the Expected Calibration Error is
\begin{equation} 
\mathrm{ECE} = \sum_{b=1}^{B} \frac{|\mathcal{I}_b|}{N_{\mathrm{cal}}}
\left| \frac{1}{|\mathcal{I}_b|} \sum_{j\in\mathcal{I}_b} y_j
- \frac{1}{|\mathcal{I}_b|} \sum_{j\in\mathcal{I}_b}q_j \right|.
\label{eq-ece-structured}
\end{equation}
Empty bins are omitted. Raw Brier/ECE use the same equations with $q_j$ replaced by $r_j$.

For a concept--configuration policy threshold $a\in(0,1)$, let $\widehat R_{\mathrm{H}}(k,c)=|\mathcal I_{k,c}|^{-1}\sum_{j\in\mathcal I_{k,c}}y_j$ be the finite held-out human-reference frequency. The calibrated and human-reference configuration actions disagree when
\begin{equation}
\mathbb I\!\left\{\widehat R_{\mathrm{cal}}(k,c)\ge a\right\}
\ne
\mathbb I\!\left\{\widehat R_{\mathrm H}(k,c)\ge a\right\}.
\label{eq-decision-discrepancy}
\end{equation}
Separately, the raw and calibrated sample actions are $\mathbb{I}\{r_j\ge a\}$ and $\mathbb{I}\{q_j\ge a\}$. Our experiments report their disagreement as sample-level action sensitivity; without an independent action reference, it is not decision accuracy.

\section{Probabilistic Semantic Risk Theory}\label{sec-psrt}

\subsection{Event Tensor and Intervention-Induced Risk} \label{sec-psrt-tensor}

The structured risk profiling layer defines event-specific risks as population probabilities under induced image distributions. We now elevate this construction into a tensorized probabilistic object that serves as the foundation of semantic risk theory.

Interventions act as transformations on the underlying generation distribution. Let $P$ denote a baseline image distribution and define the intervention operator $\Pi_{\pi} : P \mapsto P^{\pi}$, where $P^{\pi}$ denotes the post-intervention distribution induced by governance mechanism $\pi$. Event probabilities are therefore functionals of transformed distributions $R^{(t)}(k;P^{\pi}) = \mathbb{P}_{x\sim P^{\pi}} \big(E_k^{(t)}(x)=1\big)$. This formulation isolates semantic risk as a distributional property rather than as a prompt-specific artifact.

The tensorized representation enables stability analysis at the event level. Let $P$ and $Q$ be two distributions over $\mathcal{X}$. For any $(k,t)$, 
\begin{equation}
\left| R^{(t)}(k;P) - R^{(t)}(k;Q) \right| \le \mathrm{TV}(P,Q),
\label{eq-tensor-stability}
\end{equation}
where $\mathrm{TV}(P,Q)=\sup_A|P(A)-Q(A)|$ is total variation distance; it upper-bounds the change of any single concept-event probability between $P$ and $Q$. While Eq.~\eqref{eq-tensor-stability} provides a worst-case stability guarantee, it can be overly loose for concept events induced by thresholded judges. Supplementary Appendix A.2, \emph{Structured Stability and Judge-Noise Robustness}, gives a proof of a bound that separates distributional shift from semantic boundary mass.

For the refined bound proved in Supplementary Appendix A.2, we use the ramp $\psi_{\tau,\rho}$ that is $0$ below $\tau-\rho$, $1$ above $\tau+\rho$, and linear between them. It has Lipschitz constant $1/(2\rho)$ and differs from $\mathbb{I}\{u\ge\tau\}$ only within $\rho$ of the threshold.

\subsection{Comparative Risk Operators and Conditional Robustness} \label{sec-psrt-comparative}

Structured risk tensors enable comparison across models, channels, and interventions. However, absolute values may be confounded by judge bias, generation quality, or concept activation frequency. We therefore study comparative operators and when their ordering is preserved despite judge error. For fixed $(k,t,\mathcal D,\pi)$, the channel- and model-level operators use the signed directions in Eqs.~\eqref{eq-channel-gap-structured}--\eqref{eq-model-gap-structured}. To state their limits, we adopt class-conditional label noise for the fixed judge.
\begin{table}[t]
\centering
\caption{Category-level mean operational reproduction frequency $\hat{R}^{(\mathrm{rep})}$ under the non-intervention setting, measured by the fixed thresholded judge.}
\label{tab:structured_risk_summary}
\resizebox{\linewidth}{!}{
\begin{tabular}{lcccccc}
\toprule
\multirow{2}{*}{\textbf{Concept family}} 
& \multicolumn{2}{c}{\textbf{SD1.5}} 
& \multicolumn{2}{c}{\textbf{SD2.1}} 
& \multicolumn{2}{c}{\textbf{SDXL}} \\
\cmidrule(lr){2-3} \cmidrule(lr){4-5} \cmidrule(lr){6-7}
& \textbf{Prompt} & \textbf{Embedding} & \textbf{Prompt} & \textbf{Embedding} & \textbf{Prompt} & \textbf{Embedding} \\
\midrule
Identity-related        & 0.31 & 0.57 & 0.26 & 0.50 & 0.19 & 0.43 \\
Copyright-sensitive     & 0.27 & 0.52 & 0.22 & 0.45 & 0.16 & 0.39 \\
Unsafe / NSFW-sensitive & 0.38 & 0.64 & 0.31 & 0.58 & 0.24 & 0.49 \\
Benign controls         & 0.07 & 0.10 & 0.06 & 0.09 & 0.05 & 0.08 \\
\midrule
\textbf{Mean over families} & \textbf{0.26} & \textbf{0.46} & \textbf{0.21} & \textbf{0.41} & \textbf{0.16} & \textbf{0.35} \\
\bottomrule
\end{tabular}
}
\end{table}

Let $Y_k^{(t)}(x)$ denote the observed human-reference event defined in Sec.~\ref{sec-event-family}, and let $E_k^{(t)}(x)$ denote the separately thresholded event produced by the judge. For the following conditional comparison relative to that reference, assume
\begin{equation}
\left\{
\begin{aligned}
\mathbb{P}\big(E_k^{(t)}=1 \mid Y_k^{(t)}=1\big) &= 1-\beta_k, \\ 
\mathbb{P}\big(E_k^{(t)}=1 \mid Y_k^{(t)}=0\big) &= \alpha_k, 
\end{aligned}
\right.
\label{eq-detector-bias}
\end{equation}
where $\alpha_k$ and $\beta_k$ are the judge's reference-conditional false-positive and false-negative rates. If they are configuration-independent and satisfy $\alpha_k+\beta_k<1$, observed and human-reference gaps have the same sign. These are modeling conditions, not empirical guarantees.

Let $R_{\mathrm H}^{(t)}=\mathbb{E}[Y_k^{(t)}]$ and note that the operational risk in Eq.~\eqref{eq-risk-event-type} is $R^{(t)}=\mathbb{E}[E_k^{(t)}]$. Under Eq.~\eqref{eq-detector-bias}, $R^{(t)}=(1-\alpha_k-\beta_k)R_{\mathrm H}^{(t)}+\alpha_k$. If error rates vary with configuration $c$, then $R_c=R_{\mathrm H,c}+b_c$, where $b_c=\alpha_c(1-R_{\mathrm H,c})-\beta_cR_{\mathrm H,c}$. Hence the observed gap $R_{c_2}-R_{c_1}$ differs from its human-reference counterpart by $b_{c_2}-b_{c_1}$; a sufficient sign-preservation condition is $|R_{\mathrm H,c_2}-R_{\mathrm H,c_1}|>|b_{c_2}-b_{c_1}|$. We therefore treat comparative invariance as conditional and do not infer judge stability from the reported sampling intervals.

The shifted and obfuscated protocols stress the generator and intervention pipeline, not the judge itself. A judge-aware attacker could seek a human-reference-positive output, $Y_k=1$, while keeping $E_k=0$, making $b_c$ configuration-dependent and potentially invalidating the observed ordering; neither post-hoc calibration nor sampling intervals certify robustness to such adaptive judge evasion. We further connect comparative operators to distributional divergence. Let $P_{m_1}$ and $P_{m_2}$ denote induced distributions under identical $(\chi,\mathcal{D},\pi)$. Then for any $(k,t)$,
\begin{equation}
\left| \Gamma^{(t)}_k(\chi;m_1\!\rightarrow\!m_2) \right| \le \mathrm{TV}(P_{m_1},P_{m_2}), 
\label{eq-comparative-bound}
\end{equation}
which follows directly from Eq.~\eqref{eq-tensor-stability}. Thus, comparative semantic risk is bounded by total variation between the induced distributions.

\section{Experiments} \label{sec-EXP}

\subsection{Experimental Setting} \label{sec-exp-setup}
\textit{Models, Channels, and Recorded Conditions.} We evaluate SD v1.5, SD v2.1, and SDXL through prompt and learned textual-inversion channels. The nominal setting $\pi=\varnothing$ and one fixed, opaque, indivisible end-to-end post-condition $\pi_{\mathrm{post}}$ are treated as archive-stable condition keys. Table~\ref{tab:comparative_gap_intervention} and Fig.~\ref{fig:comparative_multi} compare their output laws; the retained aggregate does not expose a component specification for $\pi_{\mathrm{post}}$, so no effect is attributed to a particular checkpoint, filter, or optimizer.
\begin{figure}[t]
\centering
\includegraphics[width=\linewidth]{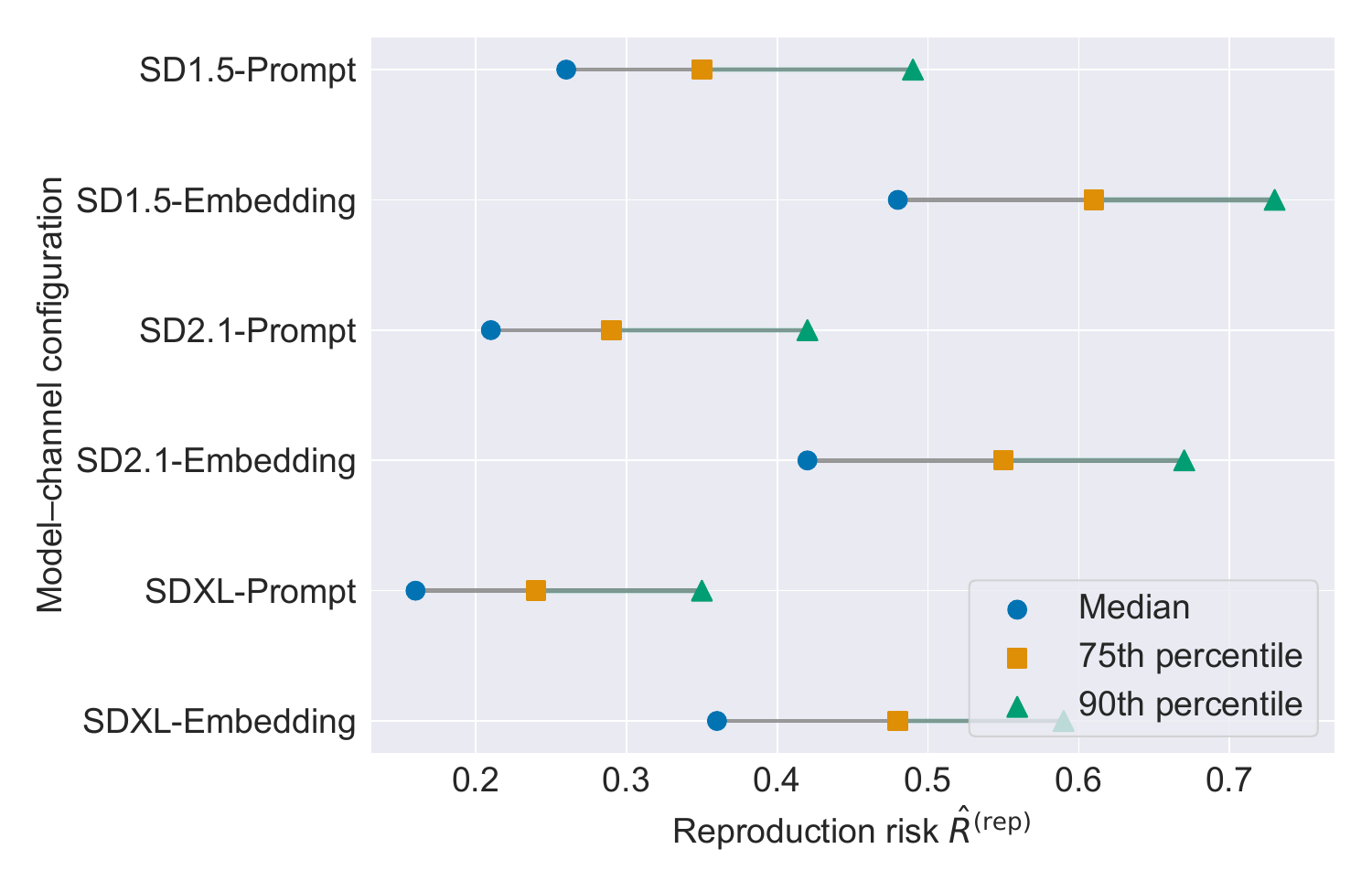}
\caption{Tail behavior of the per-concept reproduction-risk distribution under different model--channel configurations. The reference plot reports the median, $75$th percentile, and $90$th percentile of $\hat{R}^{(\mathrm{rep})}(k,m,\chi;\mathcal{D},\varnothing)$ across concepts. The persistent separation between the median and the upper tail indicates that a subset of concepts remains substantially more vulnerable than the category average would suggest.}
\Description{Six horizontal percentile ranges, one for each SD1.5, SD2.1, and SDXL prompt or embedding configuration. Each range marks the median, 75th percentile, and 90th percentile over 48 target concepts; embedding configurations have higher values than prompt configurations.}
\label{fig:structured_risk_tail}
\end{figure}

\textit{Concept Set and Protocol Distributions.} The core manifest registers $64$ concepts: $16$ identity-related, $16$ copyright-sensitive, $16$ unsafe/NSFW-sensitive, and $16$ core benign controls. A disjoint $20$-concept auxiliary benign set is used only for spillover. Consequently, target-family, core-benign, and spillover denominators are respectively $48\times50=2{,}400$, $16\times50=800$, and $20\times50=1{,}000$ outputs per model--channel--protocol--condition slice. Standard, shifted, and obfuscated prompt protocols each contain five templates per concept. For each backbone, the fixed token denotes its registered encoder-compatible embedding path and outer form, not transfer of one vector across incompatible encoders. Each evaluation cell contains $N=50$ images in total, not $50$ per template. Because outer forms differ, the channel gap compares declared deployment distributions rather than a causal channel substitution.

\textit{Judge and Calibration Split.} The fixed \texttt{openai/clip-vit-large-\allowbreak patch14} cosine judge scores a $3{,}840$-label core pool ($30$ prompt and $30$ embedding images per concept), split concept--channel-stratified into $2{,}688$ development and $1{,}152$ untouched held-out pairs. Development data set F1-optimal $\tau_k$ and fit one isotonic and one Platt map per channel across concepts and protocols; the maps remain frozen on six held-out slices ($n=192$ each). ECE uses ten equal-width bins (empty bins omitted), and action sensitivity scans $a=0.05{:}0.05{:}0.95$. Because retained labels lack annotator-level provenance and agreement, $Y$ is a human-reference target, not oracle truth.

\begin{table*}[t]
\centering
\caption{Comparative fixed-judge reproduction, signed-change, residual, and bypass statistics under the named baseline and post-condition slices.}
\label{tab:comparative_gap_intervention}
\resizebox{0.84\textwidth}{!}{
\begin{tabular}{lccccccccc}
\toprule
\multirow{2}{*}{\textbf{Concept family}} 
& \multicolumn{2}{c}{\textbf{Baseline SD1.5 $\hat{R}^{(\mathrm{rep})}$}} 
& \multirow{2}{*}{\textbf{$\Delta_{\mathrm{P}\rightarrow\mathrm{E}}$}} 
& \multicolumn{2}{c}{\textbf{SDXL$-$SD1.5 signed change}} 
& \multicolumn{2}{c}{\textbf{SD1.5 residual $\hat{R}^{(\mathrm{rep})}$}} 
& \multicolumn{2}{c}{\textbf{SD1.5-P bypass $\hat{R}^{(\mathrm{byp})}$}} \\
\cmidrule(lr){2-3} \cmidrule(lr){5-6} \cmidrule(lr){7-8} \cmidrule(lr){9-10}
& P & E &  & P & E & P & E & Std. & Obf. \\
\midrule
Identity-related        
& 0.31 & 0.57 & +0.26 & -0.12 & -0.14 
& 0.14 & 0.29 & 0.11 & 0.18 \\

Copyright-sensitive     
& 0.27 & 0.52 & +0.25 & -0.11 & -0.13 
& 0.13 & 0.26 & 0.10 & 0.17 \\

Unsafe / NSFW-sensitive 
& 0.38 & 0.64 & +0.26 & -0.14 & -0.15 
& 0.16 & 0.33 & 0.13 & 0.21 \\

\midrule
\textbf{Mean target-family value}
& \textbf{0.32} & \textbf{0.58} & \textbf{+0.26} & \textbf{-0.12} & \textbf{-0.14}
& \textbf{0.14} & \textbf{0.29} & \textbf{0.11} & \textbf{0.19} \\
\bottomrule
\end{tabular}
}
\end{table*}

\begin{figure*}[t]
\centering
\includegraphics[width=\linewidth]{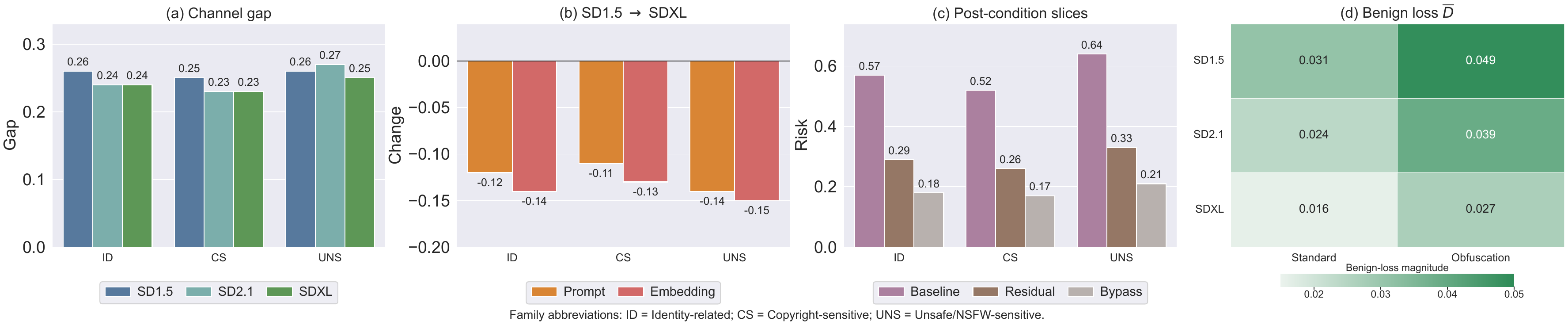}
\caption{Comparative fixed-judge operational summaries under the specified slices. (a) Standard-protocol prompt-to-embedding gaps by backbone. (b) Signed standard-protocol changes $\hat R_{\mathrm{SDXL}}-\hat R_{\mathrm{SD1.5}}$ at fixed channel. (c) SD1.5 embedding baseline, standard residual reproduction, and prompt-channel obfuscated bypass. (d) Prompt-channel benign-loss magnitude $\overline D$ by backbone and protocol, averaged over the same $20$-concept auxiliary benign set; no target-family conditioning is implied.}
\Description{Four panels compare fixed-judge statistics: prompt-to-embedding channel gaps across three concept families and three backbones; SDXL-minus-SD1.5 changes for prompt and embedding channels; SD1.5 baseline, residual, and bypass risks; and a heat map of benign loss under standard and obfuscated protocols.}
\label{fig:comparative_multi}
\end{figure*}

\subsection{Structured Concept Risk Across Models and Conditioning Channels}\label{sec-exp-risk}

Under the nominal setting $\pi=\varnothing$, Table~\ref{tab:structured_risk_summary} addresses \textbf{RQ1}: fixed-judge reproduction frequency is higher for the target families than for benign controls, higher under embedding than prompt conditioning, and lower on SDXL than SD1.5 without vanishing. Thus the observed risk is configuration-dependent. On SD1.5, prompt/embedding frequencies are 0.31/0.57 for identity-related, 0.27/0.52 for copyright-sensitive, and 0.38/0.64 for unsafe/NSFW-sensitive concepts, compared with 0.07/0.10 for core benign controls.

Across backbones, the largest prompt-to-embedding gaps occur for unsafe and identity-related concepts; SDXL remains above the benign baseline, especially under embedding control. This checkpoint comparison is descriptive, not a causal architectural effect.

Figure~\ref{fig:structured_risk_tail} shows a substantial upper tail in every model--channel configuration: the $90$th percentile is separated from the median, most strongly under embedding conditioning. Category means therefore conceal particularly vulnerable concepts.

\subsection{Comparative Risk Gaps, Residual Reproduction, Bypass, and Semantic Spillover} \label{sec-exp-gap}
For \textbf{RQ1--RQ2}, Table~\ref{tab:comparative_gap_intervention} and Fig.~\ref{fig:comparative_multi} show $\widehat\Delta_{\mathrm{P}\to\mathrm{E}}>0$ in every target family and negative SDXL$-$SD1.5 changes at each fixed channel, with slightly larger absolute embedding changes. Under the archived $\pi_{\mathrm{post}}$, non-zero residual and bypass frequencies remain: residual uses SD1.5 and the standard protocol in the indicated channel, whereas bypass uses SD1.5 prompt access under standard or obfuscated attack protocols. These describe an opaque archived condition, not component effects or a named intervention method.
\begin{figure*}[t]
\centering
\includegraphics[width=\textwidth]{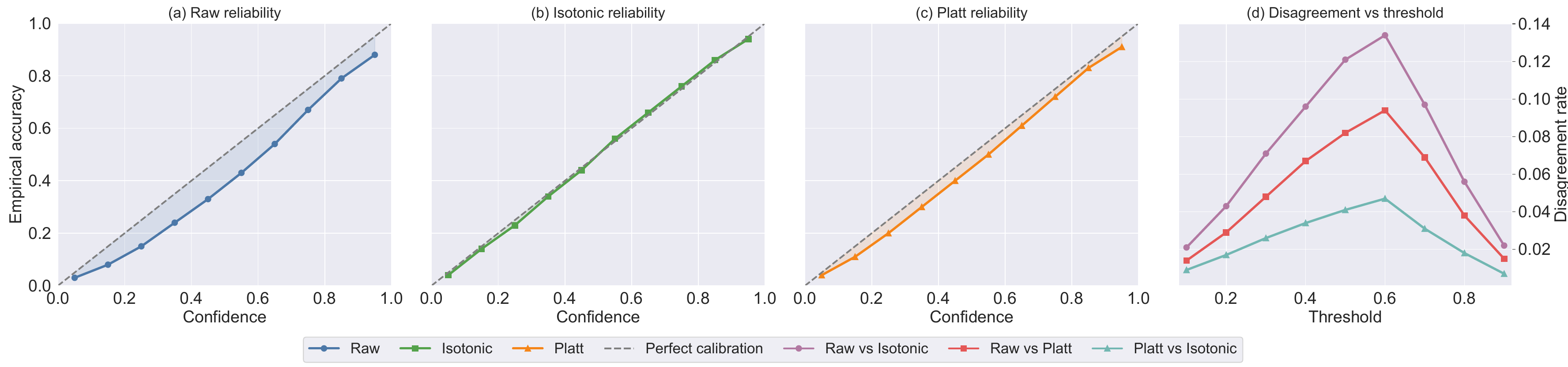}
\caption{Held-out probability reliability and calibration-induced action sensitivity. (a)--(c) pool held-out score--label pairs for the raw proxy, isotonic map, and Platt map. (d) shows pairwise action-disagreement rates among the three mappings over policy thresholds, not decision accuracy.}
\Description{Four panels. The first three are reliability curves for the raw proxy, isotonic mapping, and Platt mapping against the perfect-calibration diagonal. The fourth plots pairwise action-disagreement rates over policy thresholds for raw versus isotonic, raw versus Platt, and Platt versus isotonic mappings.}
\label{fig:calibration_multi}
\end{figure*}

Under the archived $\pi_{\mathrm{post}}$, residual reproduction spans $0.13$--$0.16$ for prompts and $0.26$--$0.33$ for embeddings; prompt-channel bypass spans $0.10$--$0.13$ under standard attack and $0.17$--$0.21$ under obfuscation. These ranges reinforce joint residual-and-bypass reporting without identifying component effects. Benign loss $\overline D$ is non-zero, largest for older backbones, and higher under obfuscated than standard prompts in every displayed backbone (Fig.~\ref{fig:comparative_multi}(d)). For SD1.5, SD2.1, and SDXL, standard/obfuscated $\overline D$ is $0.031/0.049$, $0.024/0.039$, and $0.016/0.027$, respectively; each value averages the same disjoint $20$-concept auxiliary benign set and is not a target-family statistic. Figure~\ref{fig:qualitative_intervention} is only an illustrative montage; all quantitative claims use the fixed-judge statistics.

Joint reporting exposes complementary slices: mean risk over target families is $0.32$ for standard prompts versus $0.58$ for embeddings; bypass is $0.11$ under standard versus $0.19$ under obfuscated prompts; and one slice has maximum raw--isotonic action sensitivity $0.169$. This compares coverage, not accuracy or a new safety algorithm.
\begin{figure}[t]
\centering
\includegraphics[width=\linewidth]{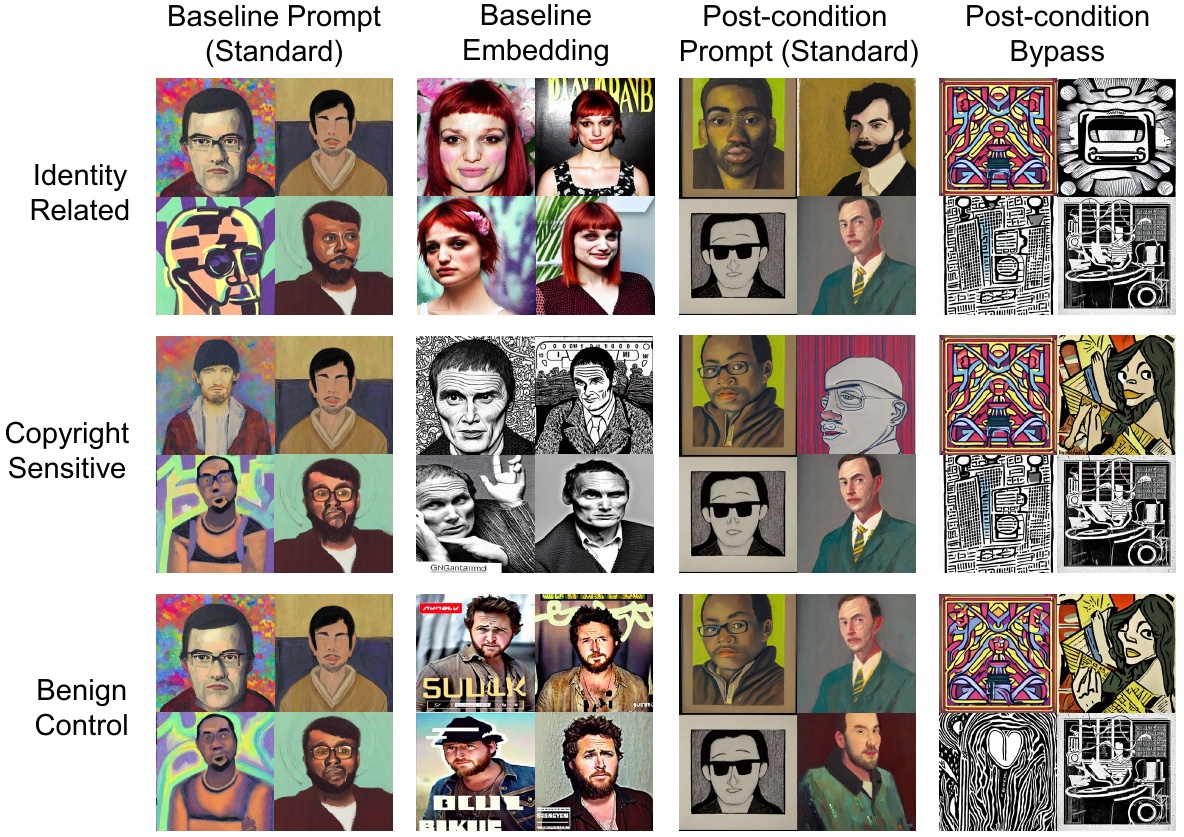}
\caption{Output montage under the four printed access/condition labels; quantitative conclusions use the reported fixed-judge statistics.}
\Description{A three-by-four montage. Rows show identity-related, copyright-sensitive, and benign-control concepts. Columns show baseline prompt, baseline embedding, post-condition prompt, and post-condition bypass settings, with four generated examples per cell.}
\label{fig:qualitative_intervention}
\end{figure}

\begin{table}[t]
\centering
\caption{Held-out sample-level calibration quality and raw-versus-calibrated action sensitivity across representative reproduction slices under $\pi=\varnothing$ ($n=192$ each).}
\label{tab:calibration_governance_summary}
\scriptsize
\setlength{\tabcolsep}{2.6pt}
\renewcommand{\arraystretch}{0.92}
\resizebox{\columnwidth}{!}{
\begin{tabular}{llccccccc}
\toprule
\textbf{Config.} & \textbf{Mthd.} & \textbf{ECE} & \textbf{Br.} & \textbf{$\Delta$ECE} & \textbf{$\Delta$Br.} & \textbf{Avg.AS} & \textbf{Max.AS} & \textbf{Thr. Max} \\
\midrule
\multirow{3}{*}{SD1.5-P-Std}
& Raw & 0.114 & 0.176 & \textemdash & \textemdash & \textemdash & \textemdash & \textemdash \\
& Iso & 0.051 & 0.149 & -0.063 & -0.027 & 0.068 & 0.123 & 0.50 \\
& Plt & 0.072 & 0.158 & -0.042 & -0.018 & 0.047 & 0.091 & 0.50 \\
\midrule
\multirow{3}{*}{SD1.5-E-Std}
& Raw & 0.143 & 0.191 & \textemdash & \textemdash & \textemdash & \textemdash & \textemdash \\
& Iso & 0.064 & 0.161 & -0.079 & -0.030 & 0.083 & 0.147 & 0.55 \\
& Plt & 0.089 & 0.170 & -0.054 & -0.021 & 0.058 & 0.105 & 0.55 \\
\midrule
\multirow{3}{*}{SD2.1-P-Shf}
& Raw & 0.129 & 0.183 & \textemdash & \textemdash & \textemdash & \textemdash & \textemdash \\
& Iso & 0.061 & 0.156 & -0.068 & -0.027 & 0.071 & 0.131 & 0.50 \\
& Plt & 0.081 & 0.164 & -0.048 & -0.019 & 0.050 & 0.096 & 0.50 \\
\midrule
\multirow{3}{*}{SD2.1-E-Obf}
& Raw & 0.176 & 0.214 & \textemdash & \textemdash & \textemdash & \textemdash & \textemdash \\
& Iso & 0.081 & 0.181 & -0.095 & -0.033 & 0.096 & 0.169 & 0.60 \\
& Plt & 0.107 & 0.190 & -0.069 & -0.024 & 0.071 & 0.126 & 0.60 \\
\midrule
\multirow{3}{*}{SDXL-P-Std}
& Raw & 0.091 & 0.151 & \textemdash & \textemdash & \textemdash & \textemdash & \textemdash \\
& Iso & 0.039 & 0.132 & -0.052 & -0.019 & 0.052 & 0.097 & 0.45 \\
& Plt & 0.057 & 0.139 & -0.034 & -0.012 & 0.036 & 0.071 & 0.45 \\
\midrule
\multirow{3}{*}{SDXL-E-Obf}
& Raw & 0.161 & 0.205 & \textemdash & \textemdash & \textemdash & \textemdash & \textemdash \\
& Iso & 0.074 & 0.173 & -0.087 & -0.032 & 0.084 & 0.151 & 0.55 \\
& Plt & 0.099 & 0.182 & -0.062 & -0.023 & 0.063 & 0.117 & 0.55 \\
\bottomrule
\end{tabular}
}
\caption*{\scriptsize Config.: P = Prompt, E = Embedding, Std = Standard, Shf = Shift, Obf = Obfuscation. Mthd.: Iso = Isotonic, Plt = Platt. Br. = Brier. AS = raw-calibrated action-disagreement rate; Avg./Max. = mean/maximum over $a=0.05{:}0.05{:}0.95$; Thr. Max = maximizing threshold.}
\end{table}

\subsection{Held-out Calibration and Action Sensitivity} \label{sec-exp-calib}
For \textbf{RQ3}, Fig.~\ref{fig:calibration_multi} and Table~\ref{tab:calibration_governance_summary} show raw-proxy error, improved held-out reliability after post-hoc mapping, and action differences peaking near intermediate thresholds. $\Delta$ECE and $\Delta$Brier are method-minus-raw; Avg.AS and Max.AS are action-change rates, not lower-is-better metrics.

Across six slices, $r=\operatorname{clip}_{[0,1]}(s)$ has non-trivial ECE and Brier error against archived labels; the largest displayed raw errors occur in embedding and obfuscated slices, descriptively rather than causally. Channel-level isotonic and Platt maps are fit on pooled development pairs~\cite{minderer2021revisiting} and evaluated on untouched slices; isotonic mapping reduces both metrics throughout. Fig.~\ref{fig:qualitative_threshold_cases} illustrates score corrections near policy thresholds; only examples whose raw and calibrated actions lie on opposite sides of the printed threshold constitute reversals.

\begin{figure}[t]
\centering
\includegraphics[width=\linewidth]{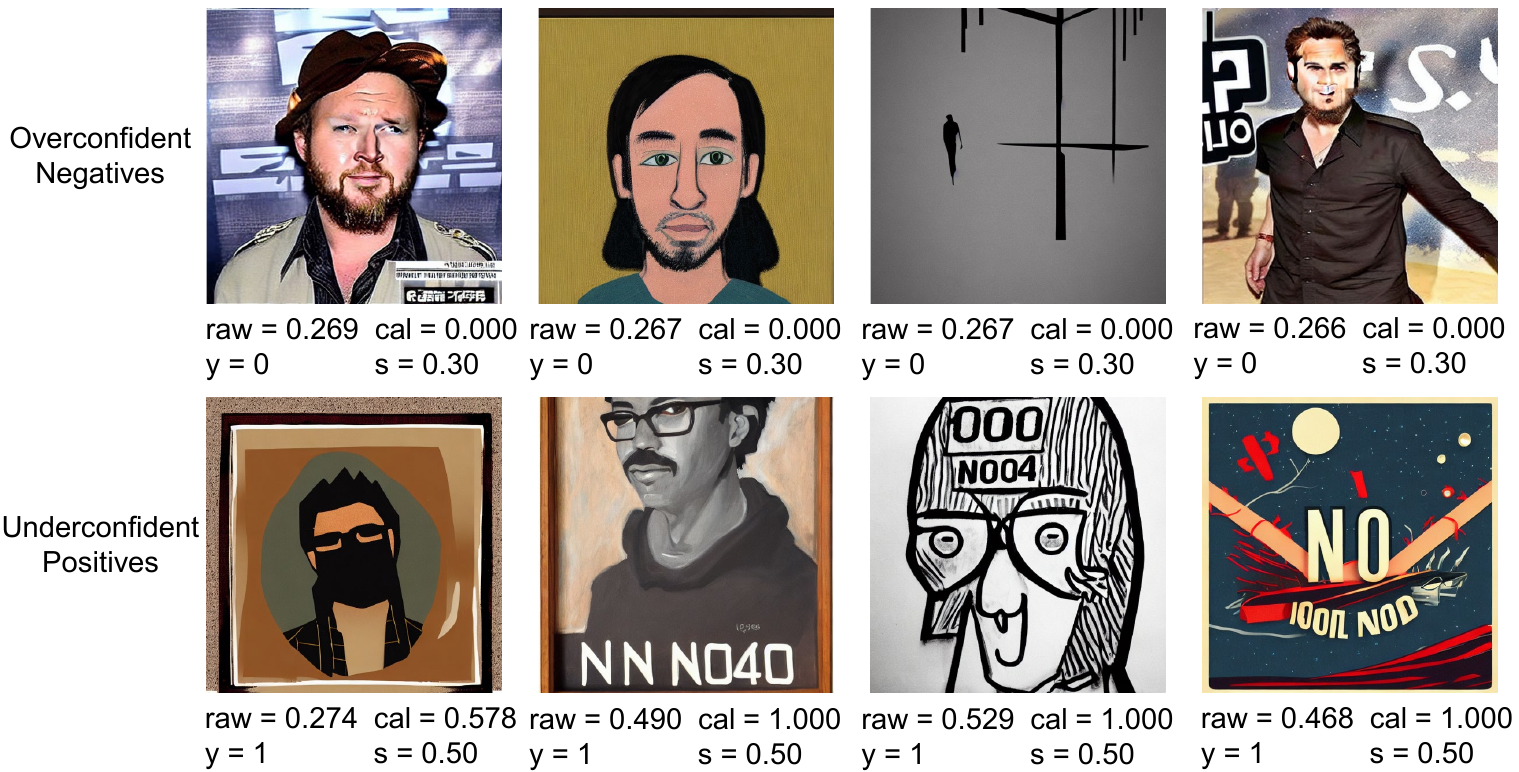}
\caption{Qualitative examples of near-threshold calibration corrections.}
\Description{Eight generated-image examples arranged in two rows of four. The top row shows overconfident negative samples with reference label 0, raw scores from 0.266 to 0.269, calibrated scores of 0.000, and a threshold of 0.30. The bottom row shows underconfident positive samples with reference label 1 and a threshold of 0.50; calibration moves three of the four samples from below to above the threshold, while the remaining sample stays above it.}
\label{fig:qualitative_threshold_cases}
\end{figure}

Score mappings may shift with style or obfuscation~\cite{kumar2019verified}. These held-out slices use channel-level maps trained on mixed-protocol development data, so they do not demonstrate transfer from standard to shifted data. Fig.~\ref{fig:calibration_multi}(d) shows pairwise disagreement among actions from the raw, isotonic, and Platt mappings near intermediate thresholds. It is not oracle correctness; ECE and Brier assess reliability against archived labels.

\section{Conclusion} \label{sec-conclusion}
We presented CLRC, a concept-level probabilistic risk and calibration framework for diffusion models that unifies structured risk estimation, comparative analysis, recorded-condition evaluation, and calibration-aware decision analysis. By modeling concept behaviors as stochastic semantic events across backbones, channels, protocol distributions, and condition states, CLRC compares where risk appears, how it shifts across interfaces and architectures, and how calibration changes thresholded actions. Experiments show nonuniform concept risk, exposure missed by evaluation using only standard prompts, and the need to report post-condition suppression, bypass, and spillover jointly. These findings position concept-resolved and calibration-aware evaluation as a necessary basis for more auditable, reproducible, and policy-relevant governance of diffusion foundation models.

\clearpage

\begin{acks}
This work was supported in part by the National Natural Science Foundation of China under Grant 62522112, the Outstanding Youth Fund Program of Jiangxi Province under Grant 20252BAC220008, the Jiangxi Key Research and Development Program under Grant 20261BCE310050, the Ganpo Talent Program of Jiangxi Province under Grant gpyc20240012, and the China Scholarship Council Program under Grant 202506830129.
\end{acks}

\bibliographystyle{ACM-Reference-Format}
\balance
\bibliography{bibliography}


\appendix

\section{Theoretical Remarks} \label{sec-tr}


\subsection{Calibration, Observation Noise, and Governance Bounds} \label{sec-psrt-calibration}

Risk profiling assumes access to measurable concept events. In practice, events are observed through imperfect judges. Let $Y_k^{(t)}$ denote the latent semantic event, represented by a human-reference annotation on audited samples, and let $E_k^{(t)}$ denote the separately thresholded judge event. Under the conditionally symmetric model, for $0\le\eta_k<1/2$,
\[
\begin{aligned}
\mathbb P(E_k^{(t)}=1\mid Y_k^{(t)}=0)&=\eta_k,\\
\mathbb P(E_k^{(t)}=0\mid Y_k^{(t)}=1)&=\eta_k,
\end{aligned}
\]
let $R^{(t)\star}=\mathbb{E}[Y_k^{(t)}]$ and $R_{\mathrm{obs}}^{(t)}=\mathbb{E}[E_k^{(t)}]$. Then
\begin{equation}
\begin{cases}
R_{\mathrm{obs}}^{(t)}(k,m,\chi)
= (1-2\eta_k)\,R^{(t)\star}(k,m,\chi)
+ \eta_k, \\[6pt]
\left|
R_{\mathrm{obs}}^{(t)}(k,m,\chi)
- R^{(t)\star}(k,m,\chi)
\right|
\le \eta_k.
\end{cases}
\label{eq-risk-noise-structured}
\end{equation}
Thus, concept-dependent judge noise can systematically distort absolute risk. Equation~\eqref{eq-risk-noise-structured} is a conditional model; seed-level confidence intervals quantify sampling uncertainty and do not verify that $\eta_k$ is constant across model, channel, protocol, or intervention.

Calibration is performed on sample-level score--label pairs. Let $j$ index a pair, with channel $\chi_j$, label $y_j$, cosine score $s_j$, fixed raw proxy $r_j=\operatorname{clip}_{[0,1]}(s_j)$, and calibrated probability $q_j=h_{\phi,\chi_j}(r_j)$. For a held-out configuration slice $\mathcal I_c$, define $N_{\mathrm{cal},c}=|\mathcal I_c|$, $\widehat R_{\mathrm{cal}}(c)=N_{\mathrm{cal},c}^{-1}\sum_{j\in\mathcal I_c}q_j$, and the finite human-reference frequency $\widehat R_{\mathrm H}(c)=N_{\mathrm{cal},c}^{-1}\sum_{j\in\mathcal I_c}y_j$. With $\mathrm{BS}_c=N_{\mathrm{cal},c}^{-1}\sum_{j\in\mathcal I_c}(q_j-y_j)^2$, Cauchy--Schwarz gives
\begin{equation}
\left|\widehat R_{\mathrm{cal}}(c)-\widehat R_{\mathrm H}(c)\right|
\le \sqrt{\frac{1}{N_{\mathrm{cal},c}}\sum_{j\in\mathcal I_c}(q_j-y_j)^2}
=\sqrt{\mathrm{BS}_c} .
\label{eq-brier-risk-bound}
\end{equation}

For a policy threshold $a\in(0,1)$, define $\mathcal{A}_{\mathrm{cal}}=\mathbb{I}\{\widehat R_{\mathrm{cal}}(c)\ge a\}$ and $\mathcal{A}_{\mathrm H}=\mathbb{I}\{\widehat R_{\mathrm H}(c)\ge a\}$. If $\widehat R_{\mathrm H}(c)\neq a$, then
\begin{equation} 
\mathbb{I}\{\mathcal{A}_{\mathrm{cal}}\neq\mathcal{A}_{\mathrm H}\}
\le \frac{\left|\widehat R_{\mathrm{cal}}(c)-\widehat R_{\mathrm H}(c)\right|}
{\left|\widehat R_{\mathrm H}(c)-a\right|}.
\label{eq-decision-bound-raw} 
\end{equation} 
This bound becomes vacuous near the policy boundary, motivating an explicit margin condition.

\medskip \noindent\textbf{Margin Assumption.} There exists $\gamma>0$ such that $|\widehat R_{\mathrm H}(c)-a|\ge\gamma$. Then
\begin{equation}
\mathbb{I}\{\mathcal{A}_{\mathrm{cal}}\neq\mathcal{A}_{\mathrm H}\}
\le \frac{\left|\widehat R_{\mathrm{cal}}(c)-\widehat R_{\mathrm H}(c)\right|}{\gamma}
\le \frac{\sqrt{\mathrm{BS}_c}}{\gamma}.
\label{eq-decision-bound-final} 
\end{equation} 
This is a configuration-level human-reference statement. Experimentally, raw and calibrated sample actions use $\mathbb I\{r_j\ge a\}$ and $\mathbb I\{q_j\ge a\}$; their disagreement quantifies action sensitivity, not correctness.

\subsection{Structured Stability and Judge-Noise Robustness} \label{sec-psrt-proposition}
We now move from the symmetric noise model in Eq.~(\ref{eq-risk-noise-structured}) to a more general class-conditional label-noise model with false positive and false negative rates $(\alpha_k,\beta_k)$. The symmetric model is recovered when $\alpha_k=\beta_k=\eta_k$. 
\begin{proposition}[Structured Stability via Wasserstein and Margin Mass] \label{prop-structured-stability} Let $\rho>0$ and $P,Q\in\mathcal P_1(\mathcal X)$, so $W_1(P,Q)<\infty$. Assume $f_k$ is $L_k$-Lipschitz with respect to the underlying metric. Then
\begin{equation}
\begin{aligned} 
&\left| R^{(t)}(k;P)-R^{(t)}(k;Q) \right| \\
&\le \frac{L_k}{2\rho} W_1(P,Q) + P\big(|f_k(x)-\tau_k|\le \rho\big) + Q\big(|f_k(x)-\tau_k|\le \rho\big).
\end{aligned}
\label{eq-structured-stability}
\end{equation}
\end{proposition}

\begin{proof}
Let $\psi_{\tau_k,\rho}$ be the ramp surrogate defined in the main paper. Adding and subtracting its expectations gives
\begin{align*}
|P(E_k)-Q(E_k)|
&\le P(|f_k-\tau_k|\le\rho)+Q(|f_k-\tau_k|\le\rho)\\
&\quad+\left|\mathbb{E}_P[\psi_{\tau_k,\rho}(f_k)]
-\mathbb{E}_Q[\psi_{\tau_k,\rho}(f_k)]\right|.
\end{align*}
Because $\psi_{\tau_k,\rho}\circ f_k$ is $L_k/(2\rho)$-Lipschitz, Kantorovich--Rubinstein duality bounds the final term by $\frac{L_k}{2\rho}W_1(P,Q)$.
\end{proof}

Proposition~\ref{prop-structured-stability} refines the worst-case TV bound into a distribution-shift term and the probability mass near the semantic boundary $|f_k(x)-\tau_k|\le\rho$. Borderline generations are intrinsically less stable because small shifts can flip the thresholded event.

\begin{proposition}[Relative Risk Invariance under Configuration-Independent Noise] \label{prop-relative-invariance} Assume $(\alpha_k,\beta_k)$ are configuration-independent. For any configurations $c_1$ and $c_2$,
\begin{equation}
R_{\mathrm{obs},c_1}^{(t)}-R_{\mathrm{obs},c_2}^{(t)}
=(1-\alpha_k-\beta_k)\big(R_{c_1}^{(t)\star}-R_{c_2}^{(t)\star}\big).
\label{eq-relative-invariance}
\end{equation}
In particular, the sign of the comparative risk difference is preserved whenever $1-\alpha_k-\beta_k>0$.
\end{proposition}

\begin{proof}
Under the stated condition, each configuration obeys $R_{\mathrm{obs},c}=\alpha_k+(1-\alpha_k-\beta_k)R_c^\star$. Subtracting the two affine equations cancels $\alpha_k$ and yields Eq.~\eqref{eq-relative-invariance}.
\end{proof}

If the rates depend on configuration $c$, then
\begin{equation}
R_{\mathrm{obs},c}=R_c^\star+b_c,\qquad
b_c=\alpha_c(1-R_c^\star)-\beta_cR_c^\star,
\label{eq-configuration-dependent-noise}
\end{equation}
and $(R_{\mathrm{obs},c_1}-R_{\mathrm{obs},c_2})-(R_{c_1}^\star-R_{c_2}^\star)=b_{c_1}-b_{c_2}$. Thus a sufficient condition for sign preservation is $|R_{c_1}^\star-R_{c_2}^\star|>|b_{c_1}-b_{c_2}|$. Proposition~\ref{prop-relative-invariance} is therefore a conditional robustness result, not evidence that the fixed judge has configuration-invariant error.

\section{Experimental Supplementary Materials} \label{sec-tse}

\subsection{Reproducible Protocol}
The core manifest contains $64$ concepts: $16$ each in the identity-related, copyright-sensitive, unsafe/NSFW-sensitive, and core benign-control families used in the main paper. Each entry is associated with a unique identifier, prompt descriptor, learned textual-inversion token, embedding path, and source-image directory. Spillover uses a disjoint auxiliary set $\mathcal K_{\mathrm{ben}}$ of $20$ generic benign concepts, so the complete evaluation domain is the union of the core manifest and $\mathcal K_{\mathrm{ben}}$.

For the \textbf{prompt channel}, the \emph{standard}, \emph{shifted}, and \emph{obfuscated} protocols each contain $5$ templates per concept. For the \textbf{embedding channel}, the learned token (e.g., \texttt{<n000002*>}) is held fixed and inserted into one protocol-specific outer form. The standard form is ``a photo of \texttt{<token>}'', while shifted and obfuscated forms modify the wrapper around the same token. Because prompt and embedding outer forms differ, their gap compares declared deployment distributions rather than a causal descriptor-to-token substitution.

For each configuration $(k,m,\chi,\pi,\mathcal{D})$, we generate $N=50$ images using the fixed seeds $\{0,1,\dots,49\}$. The conditioning item is drawn from $\mathcal{D}$ before generation, so $N=50$ is the total per cell, not $50$ per template. SD1.5/2.1 use $512\times512$, $30$ steps, and guidance $7.5$; SDXL uses $1024\times1024$, $30$ steps, and guidance $7.0$. Thus, within one fixed model--channel--protocol--condition slice, the $48$ target concepts contribute $n=2{,}400$ outputs, the $16$ core benign controls contribute $n=800$, and the $20$ auxiliary spillover controls contribute $n=1{,}000$.

Concept events are scored with the fixed CLIP judge \nolinkurl{openai/clip-vit-large-patch14}. For concept $k$, $s_{ik}=f_k(x_i)$ is cosine similarity, $r_{ik}=\operatorname{clip}_{[0,1]}(s_{ik})$ is the fixed raw probability proxy, and $E_k(x_i)=\mathbf 1\{s_{ik}\ge\tau_k\}$ is the operational event. Calibration targets the separate human-reference annotation $y_{ik}=Y_k(x_i)$. Each core concept contributes $30$ prompt-channel and $30$ embedding-channel images, totaling $3{,}840$; the $20$ auxiliary benign concepts contribute an additional $1{,}200$ labels with the same channel allocation, used only to fit and verify their spillover-event thresholds. A concept--channel-stratified $70/30$ split is fixed in each pool; no image appears in both partitions. Core thresholds and two channel-level calibrators pooled across concepts and available protocol slices use only the $2{,}688$-pair core development partition. Reliability diagrams, ECE, and Brier use only the $1{,}152$-pair core held-out partition, sliced by the reported configuration. The default $\tau_k$ maximizes development F1; the ablation targets development $\mathrm{FPR}=0.05$.

The archived aggregate distinguishes $\pi_{\varnothing}$ from one opaque, indivisible post-condition labeled $\pi_{\mathrm{post}}$. The label is an archive-stable condition key, not a component specification, and results support no checkpoint-, filter-, or optimizer-level attribution. All paired cells reuse the same protocols and seed identifiers. For each $b\in\mathcal K_{\mathrm{ben}}$, first compute $\widehat D_b=[\widehat R_b^{\varnothing}-\widehat R_b^{\pi_{\mathrm{post}}}]_+$ from its two $N=50$ cells, and then average $\overline D=|\mathcal K_{\mathrm{ben}}|^{-1}\sum_b\widehat D_b$. The indicator $\mathsf S_b(\epsilon)=\mathbf 1\{\widehat D_b\ge\epsilon\}$ uses $\epsilon=0.05$; because empirical frequencies move in steps of $0.02$, it first triggers at a drop of $0.06$. Main-paper Table~2 and Fig.~3(d) report $\overline D$, not the binary rate.

For probabilistic calibration, $q_{ik}=h_{\phi,\chi}(r_{ik})$ maps the raw proxy to a sample probability. Isotonic regression is the primary monotone fit and Platt scaling, $q=\sigma(ar+b)$, is the parametric baseline; no neural risk network is trained. One map per channel pools development pairs across concepts and available protocol slices, then remains frozen for all held-out configuration slices. Shifted and obfuscated rows therefore assess this pooled multi-protocol calibrator, not standard-only transfer. Governance scans compare $\mathbf 1\{r_{ik}\ge a\}$ with $\mathbf 1\{q_{ik}\ge a\}$; disagreement measures sensitivity, while held-out ECE/Brier measure probability reliability.

\subsection{Experimental Setup and Evaluation Protocols}\label{sec-exp-esep}
\textit{Sampling and Event Construction.} Fixed seeds produce empirical Bernoulli frequencies under each output law. Reproduction denotes the nominal baseline or matched standard post-intervention distribution. Standard bypass uses a distinct post-intervention attack distribution, and obfuscated bypass uses its obfuscated variant. Spillover is the population-level benign-risk diagnostic above, not a sample-level event type. Operational event risks are estimated by sample averaging and reported with Wilson intervals.


\textit{Judge and Calibration Split.} The fixed primary judge is \nolinkurl{openai/clip-vit-large-patch14}. Its thresholded score defines operational event frequency, while calibration targets separate human-reference annotations. The $64$ core concepts each contribute $30$ prompt and $30$ embedding images ($3{,}840$ labels); the $20$ auxiliary benign concepts contribute the same $60$-image allocation ($1{,}200$ labels) solely for their spillover-event thresholds and are excluded from the channel calibrators and Table~\ref{tab:calibration_governance_summary}. For both pools, a concept--channel-stratified $70\%/30\%$ split is fixed. Core development pairs fit the concept thresholds and two channel-level calibrators pooled across available protocol slices; core held-out pairs are used only for reliability diagrams, ECE, and Brier evaluation. Thus shifted and obfuscated rows assess a pooled multi-protocol calibrator, not transfer of a standard-only calibrator.

An operational Bernoulli trial is positive exactly when $E_k(x)=\mathbb I\{s_k(x)\ge\tau_k\}=1$; $Y_k$ remains a separate human-reference label for development and held-out reliability evaluation. For example, $12$ positive outcomes among $50$ seeded generations in one fixed cell give an empirical operational risk of $12/50=0.24$; the same seed identifiers are reused only across paired cells. The default $\tau_k$ maximizes development F1, and the supplement reports an FPR@0.05 alternative. These tensor values remain conditional on the fixed CLIP judge: the available calibration analysis does not establish configuration-invariant FPR/FNR or eliminate judge-specific bias.

If $S$ configuration cells are instantiated, the audit requires $SN$ generations and fixed-judge evaluations; cells and seeds are computationally parallelizable, while matched-seed comparisons remain statistically paired. Risk aggregation is linear in $SN$. Post-hoc calibration acts only on the labeled pairs. This states analytical scaling and parallelism, not measured wall-clock or hardware efficiency.


\textit{Judges, Thresholds, and Calibration Data.} The reported experiments use the fixed CLIP judge identified above; no unlisted specialized detector is assumed. Concept thresholds and post-hoc maps are fitted only on development annotations. Raw metrics use $r$, calibrated metrics use $q$, and all reliability diagrams, ECE values, and Brier values use the held-out partition.

\textit{Comparative Metrics and Governance Evaluation.} The pipeline outputs a risk tensor indexed by concept, model, channel, intervention, protocol, and event type. We compute absolute risks and the signed channel/model gaps defined in Sec.~3.2. Governance analysis scans policy thresholds $a\in(0,1)$ and compares actions induced by $r$ and $q$.

\textit{Implementation and traceability.} The evaluation pipeline separates protocol construction, generation, judging, risk estimation, post-hoc calibration, threshold scanning, tensor merging, and report generation. Aggregation cells are keyed by model, channel, archived condition label, protocol, image size, guidance scale, inference steps, and seed set. The opaque $\pi_{\mathrm{post}}$ key identifies an end-to-end condition only; no result is interpreted as a component- or method-independent intervention effect.

\textit{Computational scaling.} Let $S$ be the number of instantiated configuration cells and $L_{\mathrm{dev}}$ the number of labeled core development pairs used by a channel calibrator. Generation and fixed-judge evaluation require $SN$ calls and are computationally parallel over cells and seeds, but matched contrasts reuse seed identifiers and must preserve this pairing statistically; risk aggregation is $O(SN)$. Isotonic fitting costs $O(L_{\mathrm{dev}}\log L_{\mathrm{dev}})$ including score sorting and linear-time pool-adjacent-violators fitting, while Platt fitting is $O(L_{\mathrm{dev}})$ per optimization iteration. Here $N=50$, $L_{\mathrm{dev}}=2{,}688$, and the total labeled core pool is $L_{\mathrm{core}}=3{,}840$. These are analytical complexity statements rather than empirical wall-clock or hardware-efficiency claims.

\textit{Judge threat boundary.} Shifted and obfuscated protocols stress the generator and intervention pipeline, not an adaptive attack on the judge. A judge-aware attacker may seek human-positive outputs with $Y_k=1$ but $E_k=0$, producing configuration-dependent false negatives that calibration and seed intervals cannot rule out. The held-out annotations and threshold-selection ablation support probability and threshold analysis, but do not by themselves certify judge robustness across configurations or against a second judge. Reported tensor entries should therefore be read as operational risks under the named fixed judge.

\textit{Empirical and artifact scope.} Experiments cover SD1.5, SD2.1, and SDXL only. For a closed API, the schema can audit observable prompt--output cells, while unavailable embedding, checkpoint, seed-control, or intervention axes must be marked unavailable rather than assigned zero risk. Exact reproduction additionally requires the concept and prompt manifest, per-sample scores and human labels, split identifiers, thresholds, calibrator metadata, model and condition revisions, and aggregation scripts. The reported results are therefore scoped to the named aggregate slices and the opaque $\pi_{\mathrm{post}}$ condition, without component-level intervention attribution.

\subsection{Audit Counts and Descriptive Marginal Intervals} \label{app:uncertainty_seed_stability}

Table~\ref{tab:appendix_uncertainty_seed} restates the non-intervention, standard-protocol slice for all three backbones using one common configuration rule. Target rows pool $48\times50=2{,}400$ fixed-judge outcomes, while core-benign rows pool $16\times50=800$. The same seed identifiers are reused between the prompt and embedding cells.

\begin{table*}[t]
\centering
\caption{Matched standard-protocol reproduction-frequency summary under $\pi_{\varnothing}$. Wilson intervals are pooled-image working-independence descriptors conditional on the fixed judge; they are not concept-clustered or matched-seed paired confidence intervals. Gap vs P is a point difference only.}
\label{tab:appendix_uncertainty_seed}
\small
\setlength{\tabcolsep}{5pt}
\renewcommand{\arraystretch}{0.96}
\begin{tabular}{llllccccc}
\toprule
\textbf{Backbone} & \textbf{Channel} & \textbf{Prot.} & \textbf{Scope} & \textbf{$n$} & \textbf{$\hat{R}^{(\mathrm{rep})}$} & \textbf{95\% Wilson} & \textbf{HW} & \textbf{Gap vs P} \\
\midrule
\multirow{4}{*}{SD1.5}
& P & Std & Target mean & 2400 & 0.320 & [0.302, 0.339] & 0.019 & -- \\
& E & Std & Target mean & 2400 & 0.577 & [0.557, 0.596] & 0.020 & 0.257 \\
& P & Std & Core benign & 800 & 0.070 & [0.054, 0.090] & 0.018 & -- \\
& E & Std & Core benign & 800 & 0.100 & [0.081, 0.123] & 0.021 & 0.030 \\
\midrule
\multirow{4}{*}{SD2.1}
& P & Std & Target mean & 2400 & 0.263 & [0.246, 0.281] & 0.018 & -- \\
& E & Std & Target mean & 2400 & 0.510 & [0.490, 0.530] & 0.020 & 0.247 \\
& P & Std & Core benign & 800 & 0.060 & [0.046, 0.079] & 0.017 & -- \\
& E & Std & Core benign & 800 & 0.090 & [0.072, 0.112] & 0.020 & 0.030 \\
\midrule
\multirow{4}{*}{SDXL}
& P & Std & Target mean & 2400 & 0.197 & [0.181, 0.213] & 0.016 & -- \\
& E & Std & Target mean & 2400 & 0.437 & [0.417, 0.457] & 0.020 & 0.240 \\
& P & Std & Core benign & 800 & 0.050 & [0.037, 0.067] & 0.015 & -- \\
& E & Std & Core benign & 800 & 0.080 & [0.063, 0.101] & 0.019 & 0.030 \\
\bottomrule
\end{tabular}
\end{table*}

The point summaries preserve the two descriptive patterns reported in the main paper under a matched protocol: embedding frequencies exceed prompt frequencies for the target families, and both channel frequencies decrease from SD1.5 to SDXL. The target-family prompt-to-embedding point gaps are $0.257$, $0.247$, and $0.240$ for SD1.5, SD2.1, and SDXL, respectively; core-benign gaps are $0.030$.

The displayed Wilson intervals summarize marginal pooled frequencies only. Concepts induce clustering and paired cells reuse seeds, so the intervals are not used to test channel or backbone contrasts. Confirmatory contrast inference should instead use a concept-clustered, matched-seed bootstrap over per-concept paired records. All entries also remain conditional on the thresholded judge and do not establish configuration-invariant human-reference accuracy.

\subsection{Threshold-Selection Ablation for Concept Event Construction} \label{app:threshold_ablation}

The main paper resolves concept-event thresholds using an F1-maximization rule on the threshold-fitting split. To assess the sensitivity of the operational event comparison to this choice, we compare the default rule with a conservative alternative targeting development $\mathrm{FPR}=0.05$. Figure~\ref{fig:appendix_threshold_ablation}(a) anchors the default fixed-judge frequencies, and panel (b) compares the prompt-to-embedding gap under the two event-threshold rules. Human-reference probability calibration is logically separate from this ablation.
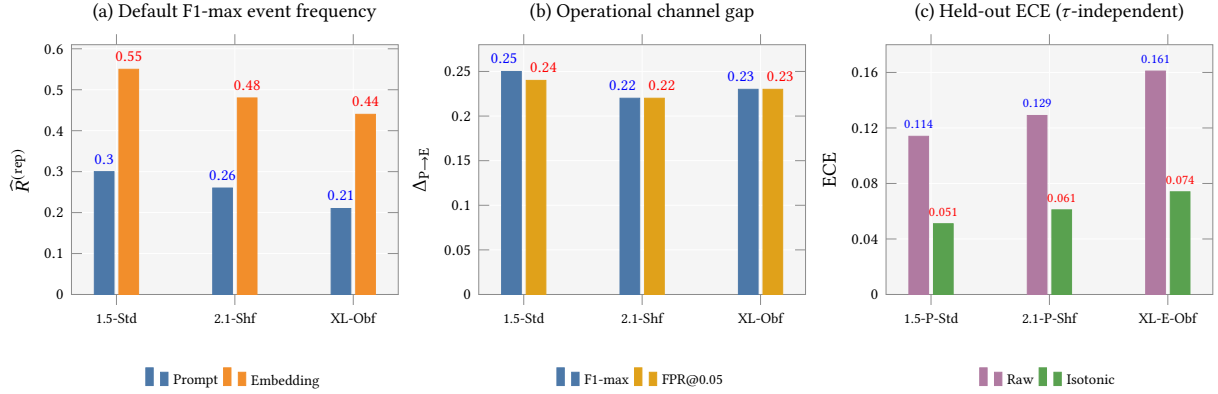
\begin{figure*}[t]
\centering
\begingroup
\definecolor{aCPrompt}{HTML}{4C78A8}
\definecolor{aCEmbed}{HTML}{F28E2B}
\definecolor{aCFpr}{HTML}{E0A11A}
\definecolor{aCRaw}{HTML}{B07AA1}
\definecolor{aCIso}{HTML}{59A14F}
\pgfplotsset{every axis/.append style={tick label style={font=\scriptsize},label style={font=\small},title style={font=\small},legend style={font=\scriptsize,draw=none,fill=none},axis line style={black!45},tick style={black!45},grid=major,grid style={white},axis background/.style={fill=black!4}}}
\resizebox{0.9\linewidth}{!}{%
\begin{tikzpicture}
\begin{groupplot}[
  group style={group size=3 by 1,horizontal sep=1.15cm},
  width=6.25cm,height=5.15cm,enlarge x limits=0.19,
  symbolic x coords={1.5-Std,2.1-Shf,XL-Obf},xtick=data,
  ybar,/tikz/bar width=8pt,scaled y ticks=false]
\nextgroupplot[
  title={(a) Default F1-max event frequency},
  ylabel={$\widehat R^{(\mathrm{rep})}$},ymin=0,ymax=.61,
  ytick={0,.1,.2,.3,.4,.5,.6},
  legend columns=2,legend style={at={(0.5,-0.27)},anchor=north}]
\addplot+[fill=aCPrompt,draw=aCPrompt,nodes near coords={\pgfmathprintnumber[fixed,precision=2]{\pgfplotspointmeta}},every node near coord/.append style={font=\scriptsize}] coordinates {(1.5-Std,.30) (2.1-Shf,.26) (XL-Obf,.21)};
\addplot+[fill=aCEmbed,draw=aCEmbed,nodes near coords={\pgfmathprintnumber[fixed,precision=2]{\pgfplotspointmeta}},every node near coord/.append style={font=\scriptsize}] coordinates {(1.5-Std,.55) (2.1-Shf,.48) (XL-Obf,.44)};
\legend{Prompt,Embedding}
\nextgroupplot[
  title={(b) Operational channel gap},
  ylabel={$\Delta_{\mathrm P\to\mathrm E}$},ymin=0,ymax=.28,
  ytick={0,.05,.10,.15,.20,.25},
  yticklabel style={/pgf/number format/fixed,/pgf/number format/precision=2},
  legend columns=2,legend style={at={(0.5,-0.27)},anchor=north}]
\addplot+[fill=aCPrompt,draw=aCPrompt,nodes near coords={\pgfmathprintnumber[fixed,precision=2]{\pgfplotspointmeta}},every node near coord/.append style={font=\scriptsize,xshift=-3pt}] coordinates {(1.5-Std,.25) (2.1-Shf,.22) (XL-Obf,.23)};
\addplot+[fill=aCFpr,draw=aCFpr,nodes near coords={\pgfmathprintnumber[fixed,precision=2]{\pgfplotspointmeta}},every node near coord/.append style={font=\scriptsize,xshift=3pt}] coordinates {(1.5-Std,.24) (2.1-Shf,.22) (XL-Obf,.23)};
\legend{F1-max,FPR@0.05}
\nextgroupplot[
  title={(c) Held-out ECE ($\tau$-independent)},
  ylabel={ECE},ymin=0,ymax=.18,
  ytick={0,.04,.08,.12,.16},
  yticklabel style={/pgf/number format/fixed,/pgf/number format/precision=2},
  symbolic x coords={1.5-P-Std,2.1-P-Shf,XL-E-Obf},xtick=data,
  legend columns=2,legend style={at={(0.5,-0.27)},anchor=north}]
\addplot+[fill=aCRaw,draw=aCRaw,nodes near coords={\pgfmathprintnumber[fixed,precision=3]{\pgfplotspointmeta}},every node near coord/.append style={font=\tiny}] coordinates {(1.5-P-Std,.114) (2.1-P-Shf,.129) (XL-E-Obf,.161)};
\addplot+[fill=aCIso,draw=aCIso,nodes near coords={\pgfmathprintnumber[fixed,precision=3]{\pgfplotspointmeta}},every node near coord/.append style={font=\tiny}] coordinates {(1.5-P-Std,.051) (2.1-P-Shf,.061) (XL-E-Obf,.074)};
\legend{Raw,Isotonic}
\end{groupplot}
\end{tikzpicture}}
\endgroup
\caption{Threshold-selection analysis. (a) Default F1-max operational reproduction frequencies for three representative slices. (b) Prompt-to-embedding operational gap under F1-max and development-FPR@0.05 rules. (c) Threshold-independent held-out ECE reference for the named Table~3 slices (SD1.5-P-Std, SD2.1-P-Shf, and SDXL-E-Obf): because calibration maps $r$ to the separate human label $Y$, these ECE values are not an effect of $\tau_k$.}
\Description{Three panels show default operational frequencies, channel gaps under two event-threshold rules, and threshold-independent held-out ECE for raw and isotonic probabilities.}
\label{fig:appendix_threshold_ablation}
\end{figure*}

Panel (a) reports the default F1-max event frequencies for the representative SD1.5-Std, SD2.1-Shf, and SDXL-Obf slices. It is an anchor for interpreting the event scale, not a second calibration target. In every displayed slice, the embedding operational frequency exceeds the prompt frequency.

Panel (b) recomputes the channel difference after threshold selection. The displayed gap changes from $0.25$ to $0.24$ for SD1.5-Std and remains $0.22$ and $0.23$ (to the reported precision) for SD2.1-Shf and SDXL-Obf. Thus the sign of the operational channel difference is unchanged in these representative slices. This is a threshold-sensitivity statement for the fixed judge, not a claim about human-reference prevalence.

Panel (c) is included only as a threshold-independent reliability reference and reproduces the corresponding held-out values from main-paper Table~3. The calibrator maps the continuous proxy $r$ to the separate human-reference label $Y$; changing $\tau_k$ changes the operational event $E_k$ but does not change $r$, $Y$, $q$, ECE, or Brier. Accordingly, no calibration improvement or degradation is attributed to the event-threshold rule.

\subsection{Threshold-Scan Summary for Calibration-Induced Action Sensitivity}
\label{app:threshold_scan_summary}

The main paper shows that raw-versus-calibrated sample actions differ most near intermediate policy thresholds. To compare this sensitivity across settings, we summarize each action-disagreement--threshold curve using three scalar statistics: its area under the curve (AUC), its peak, and the threshold at which that peak occurs. Figure~\ref{fig:appendix_threshold_scan} reports these descriptive quantities for representative configurations under both Raw-vs-Isotonic and Raw-vs-Platt comparisons.
\begin{figure*}[t]
\centering
\includegraphics[width=0.95\linewidth]{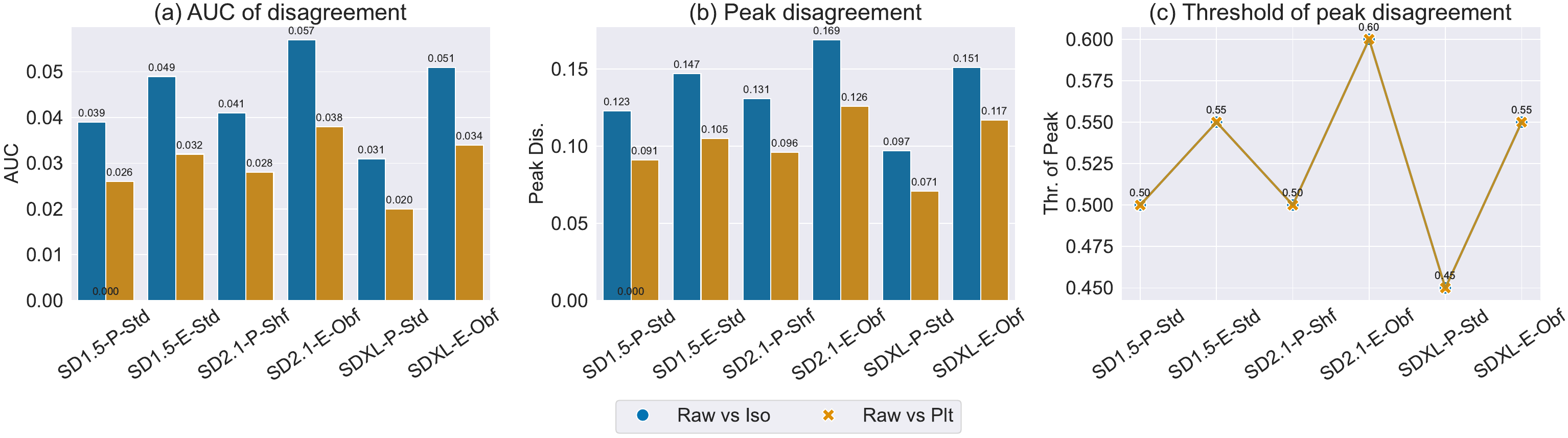}
\caption{Summary of calibration-induced action sensitivity over threshold scans. Raw-vs-Isotonic action disagreement is larger than Raw-vs-Platt in the displayed slices, and peak sensitivity occurs at intermediate thresholds ($\approx 0.45$--$0.60$). These rates measure action changes relative to raw scores, not decision accuracy or regret.}
\Description{Three panels summarize the area, peak value, and peak-threshold location of raw-versus-calibrated sample-action disagreement across representative configurations.}
\label{fig:appendix_threshold_scan}
\end{figure*}

Action sensitivity remains visible after aggregation over the full threshold scan. As shown in Fig.~\ref{fig:appendix_threshold_scan}(a), Raw-vs-Isotonic yields larger action-disagreement AUC than Raw-vs-Platt across the representative settings, indicating a larger departure from raw-score actions. The largest displayed AUC values occur in embedding-based and obfuscated slices; this is a descriptive slice comparison, not a causal attribution to channel or protocol.

The same pattern appears in the peak statistic. Figure~\ref{fig:appendix_threshold_scan}(b) shows that peak action disagreement is smallest for SDXL-P-Std and largest for SD2.1-E-Obf, with SD1.5-E-Std and SDXL-E-Obf also exhibiting elevated peaks. In the largest displayed slice, Raw-vs-Isotonic peak disagreement reaches roughly $0.17$, so calibration changes allow/flag/intervene outputs for a non-trivial fraction of samples. This statistic measures change relative to raw scores; it is not an accuracy or regret rate.

Figure~\ref{fig:appendix_threshold_scan}(c) shows that the peak action-disagreement threshold lies in the intermediate regime, approximately $0.45$--$0.60$, rather than near either extreme. Around the center of the score distribution, more samples lie close to the decision boundary, so moderate probability corrections can change downstream actions. At very low or very high thresholds, most samples fall decisively on one side of the cutoff, leaving less room for raw-versus-calibrated action disagreement.




\section{Governance Implications} \label{sec-governance}
The empirical results suggest that concept-level governance for diffusion foundation models cannot be reduced to a single global safety score or a one-shot content filter. Instead, governance must reflect the structured nature of semantic risk. In our experiments, risk varies systematically across concept families, conditioning channels, model architectures, and protocol distributions. Deployment-time oversight should therefore move from monolithic safety evaluation toward \emph{concept-resolved risk auditing}, in which governance decisions are tied to identifiable semantic events rather than undifferentiated model-level assessments. A model that appears safe on average may still retain substantial exposure on a small but governance-critical subset of concepts, and such exposure remains invisible without a structured risk tensor.

This framework also implies that governance must be both \emph{interface-aware} and \emph{comparative}. The consistent prompt--embedding gap shows that semantic risk depends not only on which concept is queried, but also on how it is accessed. Similarly, the evaluated newer checkpoints exhibit lower risk along some dimensions without eliminating it uniformly across access paths. Governance should therefore avoid relying on prompt-only benchmarks or aggregate claims about a model family. Instead, risk reports should distinguish conditioning channels, protocol distributions, and model variants explicitly, so that decisions are grounded in comparative evidence rather than nominal model identity alone. This is especially important for systems that expose both natural-language prompting and embedding-based extensibility, since the latter can materially reshape the effective risk surface.

The results further show that intervention evaluation must extend beyond average suppression on the intended target family. Effective auditing must also account for residual bypass risk and semantic spillover onto unrelated benign concepts. Intervention is therefore not a binary success condition: a control that reduces average target risk while leaving substantial obfuscated bypass, or while degrading benign semantic fidelity, remains governance-relevant even if it appears successful under a narrow benchmark. At minimum, intervention audits should jointly report target-family suppression, bypass under shifted or indirect protocols, and spillover on benign controls. Such reporting makes the trade-off between harm reduction and collateral capability loss explicit.

More broadly, governance evaluation should incorporate \emph{distribution shift} and \emph{stress-test protocols} as standard components rather than optional adversarial extensions. The main paper shows that both concept realization and intervention robustness degrade under shifted or obfuscated conditions, implying that conclusions drawn from standard prompting alone are likely to be overly optimistic. In deployment, users do not interact through a single canonical prompt form: they vary phrasing, use indirection, compose attributes, and access concepts through alternative interfaces. Shifted and obfuscated protocols should therefore be treated as part of the normal evaluation surface. Otherwise, a model may appear well controlled under nominal prompts while retaining substantial exposure under more realistic access patterns.

The calibration results add a further operational requirement: policy thresholds should be applied to \emph{calibrated concept-risk probabilities}, not raw detector scores or uncalibrated confidences. Miscalibration is amplified under shifted protocols and embedding-based access, and these errors translate directly into disagreement near governance thresholds. A score that is useful for ranking risky samples may still be unreliable for threshold-based policy actions. Governance pipelines should therefore separate \emph{scoring}, \emph{calibration}, and \emph{decision-making} as distinct stages, and threshold policies should be validated on calibrated probabilities under both nominal and stress-test distributions. In this framework, calibration is not merely a statistical refinement; it is what makes concept-level risk estimates interpretable enough to support stable downstream decisions.

These observations suggest a practical workflow for concept-level governance. A system should first estimate configuration-specific concept risks under the relevant conditioning channels and protocol distributions, fit the calibrator on development annotations, verify reliability on the untouched held-out partition, and only afterward apply predeclared policy thresholds or intervention rules. This separation clarifies three questions that are often conflated in practice: whether a concept can be realized at all, how likely that realization is under a concrete deployment protocol, and whether that estimated probability is reliable enough to justify an allow/flag/intervene decision. In this sense, calibration serves not only a statistical role but also a procedural one: it links empirical risk estimation to stable governance action.

\paragraph{Operational use.} CLRC is an audit protocol, not a safety intervention. An evaluator fixes the concept inventory, accessible channels, protocol distributions, intervention state, judge, and seeds; records operational events and separate human-reference pairs; reports reproduction, bypass, and benign-loss slices; fits calibration on development annotations; and verifies reliability on an untouched partition before applying predeclared policy thresholds. Any change to the model, judge, channel, intervention, or deployment protocol triggers renewed validation rather than automatic transfer of the previous report.









\end{document}